\documentclass[
reprint,
amsmath,
amssymb,
aps,
prx,
twocolumn,
floatfix,
superscriptaddress,
longbibliography,
preprintnumbers,
superscriptaddress,
]{revtex4-2}

\usepackage{graphicx}
\usepackage{dcolumn}
\usepackage{stix2}

\usepackage{bm}
\usepackage{braket}
\usepackage{amsthm}
\usepackage{natbib}
\usepackage{xcolor}
\usepackage{float}
\usepackage{float}
\usepackage{listings}
\usepackage{enumitem}
\usepackage{mathtools}
\usepackage{makecell,tabularx}

\usepackage[
colorlinks=true,
urlcolor=blue,
citecolor=blue,
linkcolor=blue,
hyperfootnotes=false,
pdfencoding=auto
]{hyperref}

\newtheorem{dfn}{Definition}
\newtheorem{thm}{Theorem}

\newtheorem{cor}{Corollary}

\begin{document}

\title{
High-Precision Variational Quantum SVD
via Classical Orthogonality Correction
}


\author{Shohei Miyakoshi}
\email{miyakoshi.shohei.qiqb@osaka-u.ac.jp}
\affiliation{Center for Quantum Information and Quantum Biology, Osaka University, Toyonaka, Osaka, 560-8531, Japan}
\affiliation{Computational Materials Science Research Team, RIKEN Center for Computational Science (R-CCS), Kobe, Hyogo, 650-0047, Japan}

\author{Takanori Sugimoto}
\affiliation{Center for Quantum Information and Quantum Biology, Osaka University, Toyonaka, Osaka, 560-8531, Japan}
\affiliation{Computational Materials Science Research Team, RIKEN Center for Computational Science (R-CCS), Kobe, Hyogo, 650-0047, Japan}
\affiliation{Advanced Science Research Center, Japan Atomic Energy Agency, Tokai, Ibaraki 319-1195, Japan}

\author{Tomonori Shirakawa}
\affiliation{Computational Materials Science Research Team, RIKEN Center for Computational Science (R-CCS), Kobe, Hyogo, 650-0047, Japan}
\affiliation{Quantum Computational Science Research Team, RIKEN Center for Quantum Computing (RQC), Wako, Saitama, 351-0198, Japan}
\affiliation{Computational Condensed Matter Physics Laboratory, RIKEN Cluster for Pioneering Research (CPR), Wako, Saitama, 351-0198, Japan}
\affiliation{RIKEN Interdisciplinary Theoretical and Mathematical Science Program (iTHEMS), Wako, Saitama, 351-0198, Japan}

\author{Seiji Yunoki}
\affiliation{Computational Materials Science Research Team, RIKEN Center for Computational Science (R-CCS), Kobe, Hyogo, 650-0047, Japan}
\affiliation{Quantum Computational Science Research Team, RIKEN Center for Quantum Computing (RQC), Wako, Saitama, 351-0198, Japan}
\affiliation{Computational Condensed Matter Physics Laboratory, RIKEN Cluster for Pioneering Research (CPR), Wako, Saitama, 351-0198, Japan}
\affiliation{Computational Quantum Matter Research Team, RIKEN Center for Emergent Matter Science (CEMS), Wako, Saitama 351-0198, Japan}

\author{Hiroshi Ueda}
\affiliation{Center for Quantum Information and Quantum Biology, Osaka University, Toyonaka, Osaka, 560-8531, Japan}
\affiliation{Computational Materials Science Research Team, RIKEN Center for Computational Science (R-CCS), Kobe, Hyogo, 650-0047, Japan}

\begin{abstract}
Evaluating the entanglement spectrum is essential for characterizing exotic quantum phases such as quantum criticality and topological order.
However, for large quantum many-body systems, this task is hindered by the exponential measurement complexity of standard tomographic techniques.
To address this challenge, we introduce a hybrid quantum-classical variational framework for the partial singular value decomposition of bipartite states, built upon the canonical form of matrix product states.
We employ a deflation-based optimization approach to sequentially extract the dominant and subdominant Schmidt components of target states.
Because hardware noise and finite circuit depth compromise the mutual orthogonality of these extracted vectors, we propose an \textit{improved} deflation algorithm that incorporates an explicit classical orthogonality correction.
This classical post-processing acts as an error-filtering mechanism, enabling the use of shallow and suboptimal quantum circuits.
As a result, numerical accuracy is decoupled from the quantum circuit optimization process, mitigating optimization difficulties caused by barren plateaus and hardware noise.
Furthermore, the use of shallow ansatzes enables a concurrent execution strategy.
The evaluation of overlap matrices is offloaded to classical tensor network contractions, while cross terms between the target state and the extracted vectors are computed using an auxiliary reference state.
This concurrent hybrid design improves computational throughput and bypasses the overhead of controlled target-state preparations.
Numerical benchmarks on the ground states of one- and two-dimensional Heisenberg models demonstrate improved accuracy and numerical stability.
By mitigating the hurdles of circuit depth, optimization hardness, and measurement complexity, our framework provides a robust pathway for large-scale entanglement spectrum estimation on advanced near-term devices while establishing a foundation for early fault-tolerant architectures.
\end{abstract}

\maketitle


\section{Introduction}
\label{sec:intro}

Since the advent of quantum mechanics, quantum entanglement has been extensively studied as a fundamental property distinguishing quantum systems from classical systems and as a central theme in the long-standing debate over non-locality~\cite{einstein_Can_1935,bell_einstein_1964,brunner_Bell_2014}.
In recent years, its significance has transcended the boundaries of theoretical physics and has become a widely recognized core resource that propels cutting-edge quantum technologies, including quantum computing and quantum communication~\cite{nielsen_Quantum_2010}.
This applicability extends deep into condensed matter physics, where it has emerged as a powerful probe for elucidating the characteristics of quantum many-body systems~\cite{amico_Entanglement_2008}.
Notably, it plays a pivotal role in classifying static properties such as quantum criticality~\cite{holzhey_Geometric_1994,calabrese_Entanglement_2004,tagliacozzo_Scaling_2008,pollmann_Theory_2009} and topological order~\cite{kitaev_Topological_2006,li_Entanglement_2008,pollmann_Entanglement_2010,*pollmann_Symmetry_2012,miyakoshi_Entanglement_2016}.
Beyond these static descriptions, it is also instrumental in characterizing dynamic properties such as information scrambling within quantum chaotic systems~\cite{calabrese_Evolution_2005,hosur_Chaos_2016}.
From these theoretical and applied viewpoints, 
accurately detecting and quantifying quantum entanglement on experimental platforms, such as quantum computers, is of central importance; however, the development of efficient schemes to achieve this remains a challenging task~\cite{brydges_Probing_2019,satzinger_Realizing_2021,islam_Measuring_2015,abanin_Measuring_2012}.

In light of the recent progress toward practical quantum computing, the development of precise methods for benchmarking quantum algorithms and evaluating the fidelity of quantum states generated on quantum devices is paramount~\cite{preskill_Quantum_2018}.
An important component of this endeavor is the establishment of robust methods for accurately quantifying quantum entanglement.
Current practical approaches for detecting entanglement on quantum computers include the direct measurement of $n$-th R{\'e}nyi entropy using multiple controlled-SWAP gates~\cite{ekert_Direct_2002,johri_Entanglement_2017,subasi_Entanglement_2019}, Quantum Phase Estimation (QPE)~\cite{rebentrost_Quantum_2018}, and the determination of the reduced density matrix through Quantum State Tomography (QST)~\cite{nielsen_Quantum_2010} or Classical Shadow Tomography (CST)~\cite{brydges_Probing_2019,huang_Predicting_2020,satzinger_Realizing_2021}.
However, implementing these algorithms poses significant challenges for near-term devices.
Specifically, QPE requires prohibitive circuit depths beyond the reach of Noisy Intermediate-Scale Quantum (NISQ) hardware.
Furthermore, while randomized measurement techniques such as CST have drastically reduced sampling complexity for local observables, reconstructing the full entanglement spectrum of large subsystems still demands exponential overheads relative to the subsystem size~\cite{nielsen_Quantum_2010,elben_randomized_2023}.
Variational formulations provide a promising alternative route, but they also raise a central practical challenge: how to achieve high numerical accuracy without relying on deep and highly expressive circuits that are difficult to optimize and implement on near-term hardware.
Consequently, the development of scalable and hardware-efficient entanglement detection methods is crucial for establishing reliable benchmarking techniques for quantum many-body systems on both current and future quantum computers.

To address these scalability challenges, we introduce a hybrid quantum-classical variational framework inspired by the canonical form of matrix product states (MPS)~\cite{schollwock_densitymatrix_2011}.
In the context of tensor networks (TN), imposing a canonical gauge on an MPS naturally reveals its Schmidt decomposition across a given bipartition.
Drawing on this structural property, we reformulate the extraction of the entanglement spectrum as a variational problem mapped onto parametrized quantum circuits.
Our framework assumes that a target state has already been prepared by a quantum circuit.
This state can be generated, for example, through Trotterized time evolution or a ground-state ansatz constructed variationally or rigorously for a specific Hamiltonian.
Given this prepared state, our primary objective is to construct a highly accurate partial singular value decomposition (SVD).
This involves systematically capturing both the dominant and subdominant singular values, which are essential for characterizing the underlying physical properties such as quantum criticality or topological order.

To achieve this objective beyond the expressivity limits of a single quantum circuit of constrained depth, we employ a deflation-based approach.
The core mechanism of deflation is sequential extraction.
First, a quantum circuit is optimized to capture the dominant Schmidt component, and its contribution is then subtracted from the target state.
Subsequently, the next circuit is optimized to approximate the residual.
Through this iterative optimization procedure, we systematically construct an accurate representation of the target state using a linear combination of these optimized circuits.
However, straightforward variational deflation methods often encounter numerical instabilities.
Due to limited circuit depth and inherent hardware noise, the sequentially optimized circuits do not preserve mutual orthogonality.
To resolve this limitation, we propose an \textit{improved} deflation algorithm incorporating an explicit classical orthogonality correction.
Rather than enforcing mutual orthogonality during the optimization process, this correction acts as a post-processing error-filtering mechanism on the subspace spanned by the sequentially extracted circuits.
Crucially, this filtering compensates for the limited expressivity of extracted circuits and tolerates both shallow and suboptimal ansatzes.
As a result, the final numerical accuracy no longer depends solely on the quality of the quantum circuit optimization, thereby mitigating optimization bottlenecks caused by barren plateaus and hardware noise.

Building on this resilience, the inherent shallowness of the optimized circuits enables an efficient and concurrent hybrid quantum-classical workflow.
Specifically, it allows the overlap matrices essential to the orthogonality correction to be computed accurately via classical TN contractions, thereby avoiding noise-susceptible quantum measurements.
While the classical processor handles these TN contractions, the quantum processor is concurrently dedicated to evaluating the cross-terms between the target state and the optimized circuits.
This quantum evaluation is efficiently performed by utilizing an auxiliary shallow reference state.
This strategy preserves linear signal sensitivity and bypasses the prohibitive overhead of implementing a controlled operation of the target state, thereby suppressing the error accumulation from massive gate sequences.
Rather than placing the full burden of accuracy on increasingly deep variational circuits, this architecture redistributes the computational task between classical orthogonality correction and quantum cross-term evaluation.
Accordingly, our framework provides a robust pathway for large-scale entanglement spectrum estimation, with relevance to advanced NISQ devices while establishing a foundation for early fault-tolerant architectures.
This work offers a practical route toward reconciling theoretical precision with hardware constraints, and broadens the scope of variational approaches to entanglement detection~\cite{bravo-prieto_Quantum_2020,wang_Variational_2021,kokail_Entanglement_2021,larose_Variational_2019}.

The remainder of this paper is organized as follows.
In Sec.~\ref{sec:mps_qc}, we establish the theoretical correspondence between the canonical form of the MPS and its quantum circuit representation, which serves as the foundation of our variational ansatz.
Section~\ref{sec:full_svd} introduces the variational formulation of SVD and provides a baseline demonstration using a global optimization strategy.
In Sec.~\ref{sec:partial_svd}, we present our core algorithmic contributions: the partial SVD algorithms aimed at systematically extracting dominant singular values.
We provide a detailed comparative analysis of four distinct formulations---full optimization, partial optimization, simple deflation, and improved deflation.
By benchmarking these approaches across 1D and 2D quantum spin models, we clarify the inherent limitations of the formulations lacking explicit orthogonality correction and demonstrate the enhanced numerical stability of the improved deflation algorithm.
Section~\ref{sec:discussion} discusses the practical implementation and scalability of this framework for near-term hardware.
We formulate the hybrid architecture employing the reference-state-based measurement scheme and evaluate its computational cost, highlighting the architectural synergy and concurrent execution that make the approach practical and noise-resilient.
Finally, Sec.~\ref{sec:conclusion} summarizes our findings and outlines future research directions.

\section{Theoretical Framework: Correspondence between MPS and Quantum Circuits}
\label{sec:mps_qc}


\begin{figure*}[htbp]
\centering
\includegraphics[width=1.5\columnwidth]{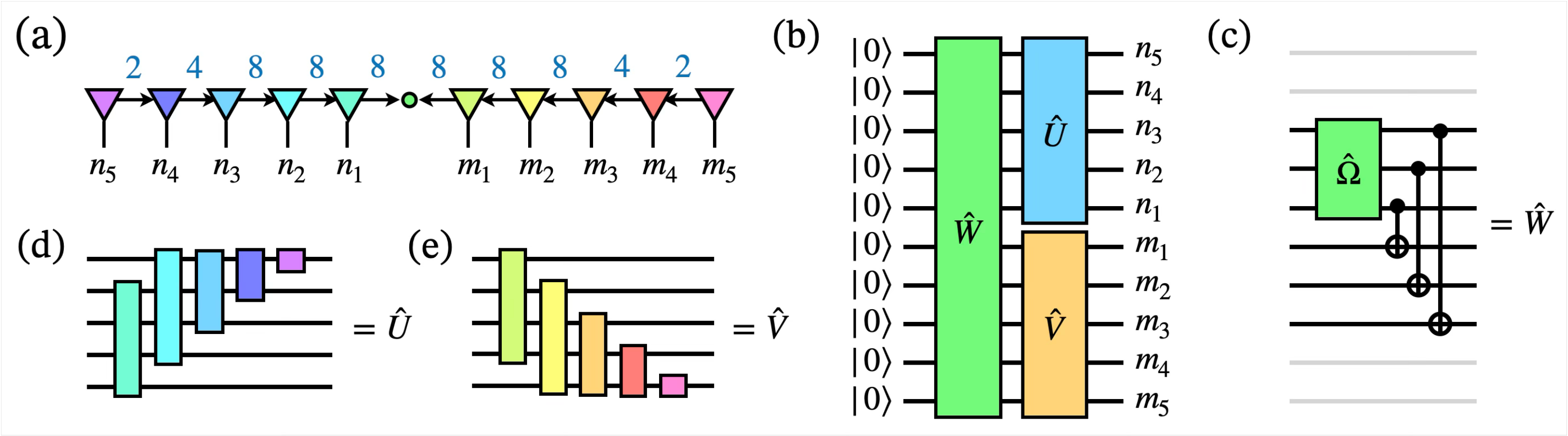}
\caption{
Correspondence between the canonical form of the MPS and quantum circuit representation.
(a) Canonical form of the MPS for an arbitrary quantum state $|\Psi\rangle$ composed of $(N+M)$ qubits. The state is bipartitioned into left and right subsystems $(A,B)$ at the boundary defined by the orthogonalized center depicted as a green circle.
The blue numbers on the bonds denote the bond dimension $D$.
(b) Quantum circuit representation of the same arbitrary state $|\Psi\rangle$, generated by applying three unitary gates: $\hat{U}$, $\hat{V}$ and $\hat{W}$, acting on the initial state $|0\rangle$.
This circuit rigorously corresponds to the MPS in (a) via Schmidt decomposition.
(c)-(e) Detailed decomposition of the unitary blocks shown in (b).
The colors of the MPS tensors in (a) correspond to the color of unitary gates in (c)-(e).
(c) The $\hat{W}$ gate, which determines the Schmidt coefficient, corresponds to the orthogonalized center in (a) and is constructed by a product of a smaller unitary gate $\hat{\Omega}$ and CNOT gates.
(d)(e) The $\hat{U}$ and $\hat{V}$ gates correspond to the unitary transformations of the SVD for the left and right subsystems, respectively.
}
\label{fig:mps-circuit}
\end{figure*}


First, we focus on the correspondence between two key representations of an arbitrary quantum state: the quantum circuit and the canonical form of the MPS.
To establish this connection, we consider an arbitrary quantum state $|\Psi\rangle$ composed of $(N+M)$ qubits and begin with its Schmidt decomposition:
\begin{align}
|\Psi\rangle=&\sum_{n,m}C_{n,m}|n,m\rangle
=\sum_{k}w_{k}|u_{k},v_{k}\rangle
\label{eq:arbitraty-state}
\end{align}
where $|n,m\rangle=|n\rangle\otimes|m\rangle$ is the direct product of the orthonormal bases $|n\rangle$ and $|m\rangle$ for the respective subsystems $A$ and $B$.
The indices $n$ and $m$ correspond to the bit-strings $n=n_{N}\cdots n_{1}$ and $m=m_{M}\cdots m_{1}$.
The coefficient matrix $C$ is expressed using SVD as $C=UWV^{T}$, where $U$ and $V$ are unitary matrices.
The diagonal elements $w_{k}>0$ of the diagonal real matrix $W$ are assumed to be in descending order.
Furthermore, the states $|u_{k}\rangle$ and $|v_{k}\rangle$ are constructed from these unitary matrices, defining the unitary gates $\hat{U}$ and $\hat{V}$ for each subsystem:
\begin{align}
|u_{k}\rangle=&
\sum_{n}|n\rangle U_{n,k}=\hat{U}|k\rangle,
\label{eq:unitary-gate-to-state1}
\\
|v_{k}\rangle=&
\sum_{m}|m\rangle V_{m,k}=\hat{V}|k\rangle.
\label{eq:unitary-gate-to-state2}
\end{align}
Next, to entangle two independent subsystems, we introduce the unitary gate $\hat{W}$ satisfying the following action for the initial state $|0\rangle$ as
\begin{align}
\hat{W}|0\rangle=\sum_{k}w_{k}|k,k\rangle.
\end{align}
where the singular values are normalized such that $\sum_{k}w_{k}^2=1$.
By combining this operator with the unitary gates $\hat{U}$ and $\hat{V}$, the arbitrary quantum state $|\Psi\rangle$ can be constructed as follows.
\begin{align}
|\Psi\rangle=(\hat{U}\otimes\hat{V})\hat{W}|0\rangle
\end{align}
The resulting quantum circuit representation for the state $|\Psi\rangle$ is shown in Fig.~\ref{fig:mps-circuit}.

As detailed in the Appendix~\ref{app:exact_quantum_circuit_representation}, the quantum circuit composed of the unitary gate $\hat{U}$, $\hat{V}$ and $\hat{W}$ corresponds rigorously to the canonical form of the MPS with the orthogonalized center, as illustrated in Fig.~\ref{fig:mps-circuit}(a).
It is worth noting that the unitary gate $\hat{U}$ and $\hat{V}$ can generally be decomposed into products of the multi-qubit gates in Fig.~\ref{fig:mps-circuit}(d)(e), where each multi-qubit gate corresponds to an MPS tensor in the left and right subsystems, respectively.
Furthermore, as shown in Fig.~\ref{fig:mps-circuit}(c), the unitary gate $\hat{W}$, which corresponds to the orthogonalized center depicted as the green circle in Fig.~\ref{fig:mps-circuit}(a), can be represented as the product of the CNOT gates and a smaller unitary gate $\hat{\Omega}$.
By tuning this unitary gate $\hat{\Omega}$, we can reproduce various singular value distributions $w_{n}$, drawing on the analogy with MPS construction described in the Appendix~\ref{app:quantum_circuit_representation_for_orthogonal_center}.
Crucially, the number of qubits $N_{q}$ shared between adjacent unitary gates determines the upper limit of the bond dimension $D$, given by $N_{q}=\log_{2}D$.
For instance, since the bond dimension between the tensors connected to the physical indices $n_1$ and $n_2$ in Fig. 1(a) is $D=8$, the number of shared qubit lines between the corresponding multi-qubit gates (the first and second from the left) in Fig. 1(d) is $N_q = \log_2(8) = 3$.

The canonical condition of the MPS tensors is inherently satisfied by the unitarity of the quantum gates.
Consequently, increasing the gate size directly enhances the expressivity of the quantum circuit, corresponding to a larger bond dimension in the MPS representation.
However, decomposing large-scale unitary gates like $\hat{U}$ and $\hat{V}$ into elementary gates is computationally challenging.
Therefore, this study focuses on optimizing an approximate quantum circuit composed of layered quantum gates with a fixed gate size of two qubits, which is readily implementable on current quantum devices.

\section{Full SVD Algorithm}
\label{sec:full_svd}

\subsection{Global Optimization Algorithm}
\label{subsec:full_svd_alg}

In this study, we utilize the established correspondence between quantum circuits and MPS to formulate the Schmidt decomposition as a variational quantum algorithm.
A distinguishing feature of our algorithm is that the gate $\hat{W}$, which determines the singular value profile, does not require optimization.
Instead, by fixing $\hat{W}$ as a reference gate defined by non-negative decreasing weights $w_{n}$, we can formulate the problem solely as a quantum circuit optimization for the unitary gates $\hat{U}$ and $\hat{V}$, which correspond to the left and right singular vectors, respectively.

To formally define this optimization task, we introduce the objective function $I(\hat{U},\hat{V})$ as the real part of the overlap between the trial state $|\Psi(\hat{U},\hat{V})\rangle$ and target state $|\Phi\rangle$.
The trial state is prepared by applying a fixed reference gate $\hat{W}$ and the variational unitary gates $\hat{U}$ and $\hat{V}$ to an initial state $|0\rangle$.
The objective function is expressed as:
\begin{align}
I(\hat{U},\hat{V})=&\mathrm{Re}\langle\Psi(\hat{U},\hat{V})|\Phi\rangle
\nonumber\\
=&\mathrm{Re}\langle 0|\hat{W}^{\dagger}(\hat{U}^{\dagger}\otimes\hat{V}^{\dagger})|\Phi\rangle\\
=&\mathrm{Re}\mathrm{tr}(WU^{\dagger}CV^{*})
\end{align}
where $C$ represents the coefficient matrix of the target state $|\Phi\rangle$, as defined in Eq.~\eqref{eq:arbitraty-state}.
The justification for maximizing $I(\hat{U},\hat{V})$ to achieve SVD is grounded in the following inequality:
\begin{align}
I(\hat{U},\hat{V})\le\sum^{D}_{n=1}w_{n}s^{(D)}_{n}
\le\sum^{D}_{n=1}w_{n}\sigma_{n}
\label{eq:inequality-vqsvd}
\end{align}
where $D=2^{\min(N,M)}$ is the maximum possible bond dimension of the bipartition, and $\sigma_{n}$ denotes the $n$-th singular value of the target state $|\Phi\rangle$.
Here, $s^{(D)}_{n}$ represents the $n$-th largest value in the set $\mathcal{S}_{D}=\{|\langle n,n|\hat{U}^{\dagger}\otimes\hat{V}^{\dagger}|\Phi\rangle|\}^{D}_{n=1}$, calculated for the optimized gates $\hat{U}$ and $\hat{V}$.
Mathematical details regarding this inequality and the definition of the set $\mathcal{S}_{D}$ are provided in Appendix~\ref{app:details_of_global_optimization}.
As this inequality suggests, provided the unitary gates $\hat{U}$ and $\hat{V}$ possess sufficient expressivity, the maximum of $I(\hat{U},\hat{V})$ converges to the weighted sum of the true singular values $\sigma_{n}$ with coefficients $w_{n}$.
In this limit, the values $s^{(D)}_{n}$ faithfully approximate the true singular values $\sigma_{n}$.

\subsection{Demonstration for a 1D Heisenberg model}
\label{subsec:full_svd_demo}


\begin{figure*}[htbp]
\centering
\includegraphics[width=1.4\columnwidth]{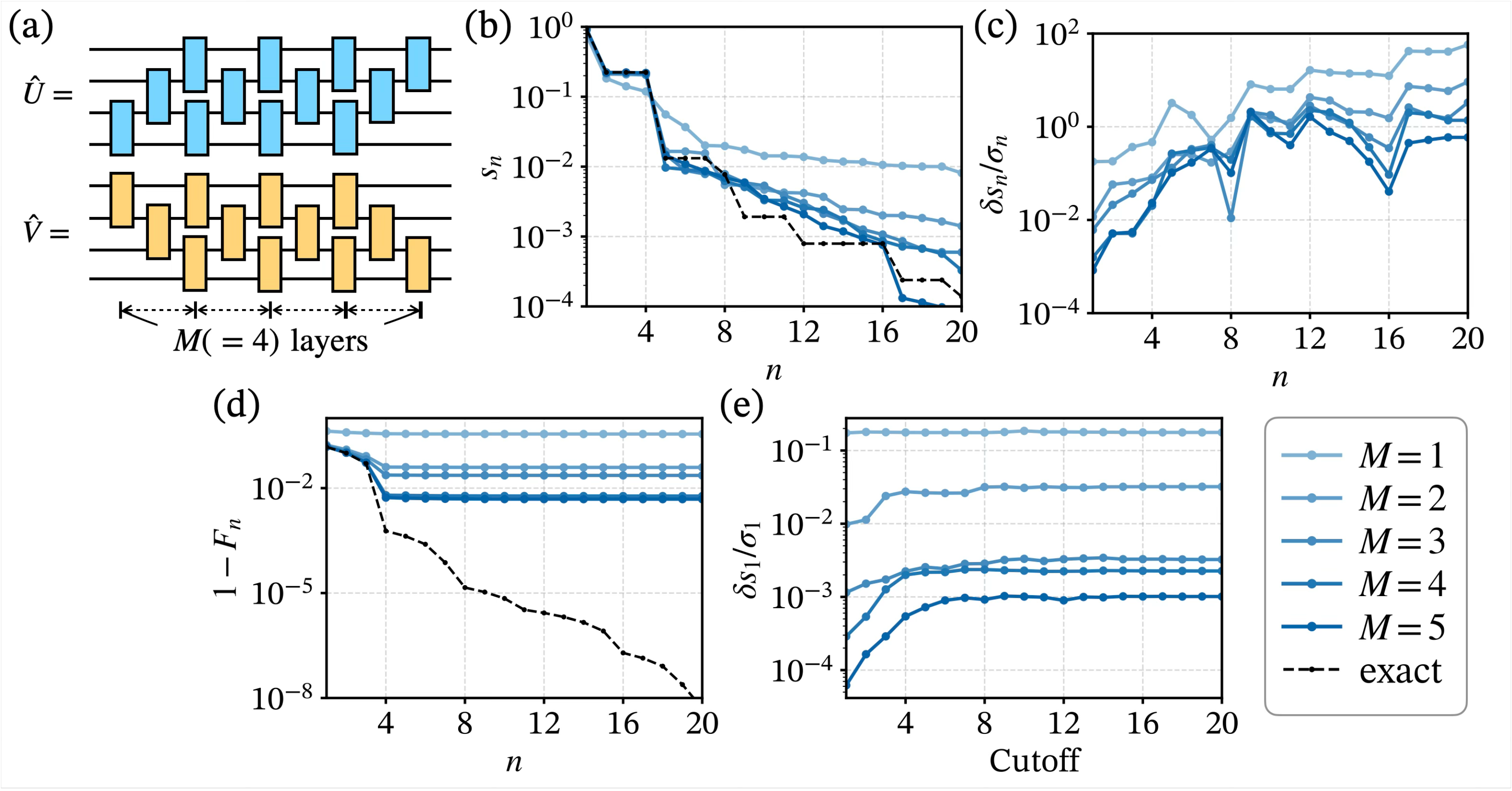}
\caption{
Numerical results of the variational Schmidt decomposition algorithm for the 1D Heisenberg model ($L=16$).
(a) Structure of the variational unitary gates $\hat{U}$ and $\hat{V}$ used for singular value estimation. Each gate is an $M$-layer quantum circuit consisting of sequentially aligned two-qubit gates.
(b) and (c) Singular value distribution $s_{n}$ and its relative error $\delta s_{n}/\sigma_{n}$, respectively.
(d) Infidelity, $1-F_{n}$, where $F_{n}=\sum^{n}_{m=1}s_{m}^2$.
(e) Relative error of the largest singular value, $\delta s_{1}/\sigma_{1}$, as a function of the cutoff.
The results in panels (b)-(e) are calculated using quantum circuits with varying depths up to $M=6$.
}
\label{fig:vqsvd-demo1}
\end{figure*}


Next, to verify the efficacy of our Schmidt decomposition method, we apply it to the ground state of the 1D Heisenberg model,
\begin{align}
\hat{H}=J\sum^{L-1}_{n=1}
\hat{\bm{S}}_{n}\cdot\hat{\bm{S}}_{n+1}
\end{align}
where $\hat{\bm{S}}_{n}$ is the spin-$1/2$ operator acting on the $n$-th qubit.
Here, we assume the system size $L=16$ and the coupling constant $J=1$ in our calculation.
To maximize the objective function $I(\hat{U},\hat{V})$, we employ a sequential sweep update algorithm for the local unitary gates constituting the circuit $\hat{G}=\hat{\mathcal{G}}_{N}\cdots\hat{\mathcal{G}}_{1}$ \cite{shirakawa_Automatic_2024,miyakoshi_Diamondshaped_2024}.
This algorithm is based on the iterative optimization of each local gate $\hat{\mathcal{G}}_{n}$ by considering the metric tensor $\hat{\mathcal{F}}_{n}$,
\begin{align}
\hat{\mathcal{F}}_{n}=\mathrm{Tr}_{\bar{n}}[\hat{\mathcal{G}}_{n+1}^{\dagger}\cdots|\Phi\rangle\langle 0|\hat{\mathcal{G}}_{1}^{\dagger}\cdots\hat{\mathcal{G}}_{n-1}]
\end{align}
where $\mathrm{Tr}_{\bar{n}}$ denotes the trace over the subspace complementary to the space upon which $\hat{\mathcal{G}}_{n}$ acts. The local update that maximizes the global objective function is derived from the SVD of this tensor, $\hat{\mathcal{F}}_{n}=\hat{X}\hat{D}\hat{Y}^{\dagger}$, yielding the optimal gate $\hat{\mathcal{G}}_{n}\leftarrow\hat{Y}\hat{X}^{\dagger}$.
Additionally, we assume that the unitary gates $\hat{U}$ and $\hat{V}$ are each composed of an $M$-layer quantum circuit  shown in Fig.~\ref{fig:vqsvd-demo1}(a).
This circuit consists of sequentially aligned 2-qubit gates applied first within subsystem $A$, and then within subsystem $B$.
The sweep update was repeated until the relative change in the objective function per sweep fell below $10^{-12}$, or until a maximum of $10^{5}$ sweep updates was reached.
The coefficients $w_{n}$ of the reference gate $\hat{W}$ follow an exponential decay defined by $w_{n}=pw_{n-1}$ with $p=0.9$, where the prefactor $w_{1}$ is determined by the normalization condition $\sum^{D}_{n=1}w_{n}^2=1$.
Here, $p$ is chosen to be close to $1$, so that the optimization remains sensitive to the entire spectrum, while still being smaller than $1$ to lift the degeneracy among the target singular vectors. The converged results are not sensitive to its precise value in this regime.

First, as a benchmark of our method's basic performance, we show the singular value distribution $s_{n}=s^{(D)}_{n}$ and its relative error $\delta s_{n}/\sigma_{n}=|1-s_{n}/\sigma_{n}|$ calculated for each quantum circuit with $M$ layers in Figs.~\ref{fig:vqsvd-demo1}(b) and (c), respectively.
The figures illustrate that the singular values $s_{n}$ approach the exact values $\sigma_{n}$ as $M$ increases. This trend is particularly prominent for the larger singular values.
Next, to evaluate the accuracy of the reproduced singular values, we plot the infidelity $1-F_{n}$ in Fig.~\ref{fig:vqsvd-demo1}(d), where $F_{n}=\sum^{n}_{m=1}(s_{m})^{2}$ is the cumulative sum of the squared singular values.
The fidelity $F_{n}$ serves as an indicator of the quantum circuit's expressivity; ideally, $F_{n}$ should converge to 1, corresponding to an infidelity of $0$.
As shown in the figure, the plotted infidelity $1-F_{n}$ saturates to a constant value as $n$ increases.
Crucially, as the circuit depth $M$ increases, the enhanced expressivity of the quantum circuit drives $F_{n}$ closer to 1, causing this saturation value of $1-F_{n}$ to approach $0$.
This successfully demonstrates that highly expressive unitary gates $\hat{U}$ and $\hat{V}$ improve the accuracy of the individual singular values and reduce the off-diagonal overlaps $|\langle k,l|\hat{U}^{\dagger}\otimes\hat{V}^{\dagger}|\Phi\rangle|$.
However, because this method targets the entire spectrum simultaneously,
 there is a concern that the circuit's representational power is distributed too broadly across all singular values, rather than being concentrated on the most dominant ones.
We therefore introduce a cutoff $c$ to the coefficients $w_{n}$, renormalize the truncated coefficients as $w_{n}\leftarrow w_{n}/\sqrt{\sum^{c}_{m=1}w_{m}^2}$, and perform a refined optimization.
Figure~\ref{fig:vqsvd-demo1}(e) shows the relative error of the largest singular value, $\delta s_{1}/\sigma_{1}$, as a function of cutoff.
As can be seen from the figure, reducing the cutoff brings the largest singular value closer to its exact value. This result suggests that in situations with limited gate depth, such as on current NISQ devices, the difficulty of representing the entire singular vector space with unitary gates becomes evident.
By restricting the target singular vectors via the cutoff, the precision of individual singular values can be enhanced.

Furthermore, as these results suggest, simply approximating the unitary matrices for the SVD as quantum circuits becomes challenging as the ground state complexity increases, a trend further demonstrated for the 2D system in Fig.~\ref{fig:vqsvd-demo3}(e).
This difficulty is equivalent to the limitation in circuit expressivity in previous studies on variational algorithms for Schmidt decomposition and SVD~\cite{larose_Variational_2019,wang_Variational_2021}, which typically aim to approximate the entire set of singular vectors.
While this approach may suffice for estimating averaged properties like entanglement entropy, it poses a significant problem for accurately estimating individual singular values.
In particular, the deep quantum circuits required for improving the accuracy of their variational quantum algorithms pose severe challenges to the optimization process and lead to error amplification on current quantum hardware.

\section{Partial SVD Algorithms}
\label{sec:partial_svd}


\newlength{\colWone}
\setlength{\colWone}{0.18\textwidth}
\newlength{\colWtwo}
\setlength{\colWtwo}{0.31\textwidth}

\begin{table*}[htbp]
\caption{Summary of the four variational SVD algorithms considered in this study.}
\label{tab:method-summary}
\renewcommand{\arraystretch}{1.5}
\begin{tabular}{|l|l|l|l|}
\hline
\textbf{Method}&
\textbf{Concept}&
\textbf{Objective Function}&
\textbf{Singular Values}\\
\hline
\hline

full&
\begin{minipage}[t]{\colWone}\flushleft
  Single Global Optimization
\vspace{1mm}
\end{minipage}&
\begin{minipage}[t]{\colWtwo}\flushleft
  Maximize the global overlap: \newline
  $I(D) = \sum^{D}_{n=1}w_{n}\mathrm{Re}\langle n,n|\hat{U}^{\dagger}\otimes\hat{V}^{\dagger}|\Psi\rangle$. \newline
  \textbf{Target:} Global gates $\hat{U}, \hat{V}$.
\vspace{1mm}
\end{minipage}&
\begin{minipage}[t]{\colWtwo}\flushleft
  Estimated as the $n$-th largest element of the diagonal overlaps: \newline
  $s_{n}=s^{[D]}_{n} \in \{|\langle k,k|\hat{U}^{\dagger}\otimes\hat{V}^{\dagger}|\Psi\rangle|\}_{k=1}^{D}$.
\vspace{1mm}
\end{minipage}\\
\hline

partial&
\begin{minipage}[t]{\colWone}\flushleft
  Multiple Global Optimizations \newline
  + Difference
\vspace{1mm}
\end{minipage}&
\begin{minipage}[t]{\colWtwo}\flushleft
  Maximize the global overlaps: \newline
  $I(c) = \sum^{c}_{n=1}w_{n}\mathrm{Re}\langle n,n|\hat{U}^{\dagger}\otimes\hat{V}^{\dagger}|\Psi\rangle$ \newline
  for two distinct cutoffs $c$ and $c-1$. \newline
  \textbf{Target:} Global gates $\hat{U}, \hat{V}$.
\vspace{1mm}
\end{minipage}&
\begin{minipage}[t]{\colWtwo}\flushleft
  Estimated via the difference of the cumulative sums of $s^{[c]}_{n}$: \newline
  $s_{c}=\sum^{c}_{n=1} s^{[c]}_{n}-\sum^{c-1}_{n=1} s^{[c-1]}_{n}$, \newline
  for two distinct cutoffs $c$ and $c-1$.
\vspace{1mm}
\end{minipage}\\
\hline

simple&
\begin{minipage}[t]{\colWone}\flushleft
  Sequential Local Optimizations \newline
  + Deflation
\vspace{1mm}
\end{minipage}&
\begin{minipage}[t]{\colWtwo}\flushleft
  Maximize the local overlap with the projected state $|\Psi_{n-1}\rangle \propto (\hat{I}-\hat{P}_{n-1})|\Psi\rangle$: \newline
  $I_{n}=\mathrm{Re}\langle u_{n},v_{n}|\Psi_{n-1}\rangle$, \newline
  using the \textbf{naive projector} \newline
  $\hat{P}_{n-1}=\sum_{m=1}^{n-1}|u_{m},v_{m}\rangle\langle u_{m},v_{m}|$. \newline
  \textbf{Target:} Local gates $\hat{U}_n, \hat{V}_n$.
\vspace{1mm}
\end{minipage}&
\begin{minipage}[t]{\colWtwo}\flushleft
  Estimated as the positive part of the real overlap: \newline
  $s_{n}=\max(\mathrm{Re}\langle u_{n},v_{n}|\Psi\rangle, 0)$, \newline
  assuming naive mutual orthogonality: \newline
  $\langle u_{n}|u_{m}\rangle=\langle v_{n}|v_{m}\rangle=\delta_{n,m}$.
\vspace{1mm}
\end{minipage}\\
\hline

improved&
\begin{minipage}[t]{\colWone}\flushleft
  Sequential Local Optimizations \newline
  + Deflation \newline
  + Orthogonality Correction
\vspace{1mm}
\end{minipage}&
\begin{minipage}[t]{\colWtwo}\flushleft
  Maximize the local overlap with the projected state $|\Psi_{n-1}\rangle \propto (\hat{I}-\hat{P}_{n-1})|\Psi\rangle$: \newline
  $I_{n}=\mathrm{Re}\langle u_{n},v_{n}|\Psi_{n-1}\rangle$, \newline
  using the \textbf{corrected projector} \newline
  $\hat{P}_{n-1}=\hat{P}^{A}_{n-1}\otimes\hat{P}^{B}_{n-1}$. \newline
  \textbf{Target:} Local gates $\hat{U}_n, \hat{V}_n$.
\vspace{1mm}
\end{minipage}&
\begin{minipage}[t]{\colWtwo}\flushleft
  Estimated from the SVD of the reduced and corrected overlap matrix: \newline
  $\Sigma^{(n)}_{k,l}=\langle\chi^{A}_{k},\chi^{B}_{l}|\Psi\rangle$, \newline
  where $|\chi^{A}_{k}\rangle$ and $|\chi^{B}_{l}\rangle$ are eigenvectors of $\hat{P}^{A}_{n-1}$ and $\hat{P}^{B}_{n-1}$, respectively.
\vspace{1mm}
\end{minipage}\\
\hline

\end{tabular}
\end{table*}


As discussed in Sec.~\ref{sec:full_svd}, variational algorithms that attempt to approximate the entire set of singular vectors simultaneously face severe expressivity bottlenecks.
The deep quantum circuits required for such global optimization inevitably lead to significant optimization difficulties and error amplification on near-term quantum hardware.
To overcome these inherent limitations of the full SVD method, we propose partial SVD algorithms aimed at accurately and efficiently extracting the dominant singular values.
By shifting the optimization target from the entire spectrum to specific dominant components, these approaches significantly reduce circuit depth and improve numerical accuracy.

To provide a clear roadmap, Table~\ref{tab:method-summary} summarizes the four variational SVD frameworks evaluated in this study, outlining their core concepts, objective functions, and the specific estimators used to extract singular values.
In the following subsections, we systematically formulate these methods and benchmark their performance.

\subsection{Full and Partial Optimization}
\label{subsec:partial_svd_alg1}

First, we define an estimator for the $k$-th singular value, $s_{k}$, as the difference between the cumulative sums obtained from two separate optimizations performed with different cutoffs, $k$ and $k-1$:
\begin{align}
s_{k} = \left(\sum^{k}_{n=1}(s^{(k)}_{n})^{p}
-\sum^{k-1}_{n=1}(s^{(k-1)}_{n})^{p}\right)^{1/p},
\label{eq:partial_svd}
\end{align}
where $s^{(k)}_{n}$ denotes the $n$-th singular value optimized with cutoff $k$ following the notation introduced in the previous section.
Note that the term inside the parentheses in Eq.~\eqref{eq:partial_svd} can become non-positive due to numerical fluctuations in the optimization.
In such cases, we treat $s_{k}$ as undefined.
For numerical robustness, we choose $p=1$, although other variants are discussed in Appendix~\ref{app:details_of_partial_optimization}.
As suggested by our previous results in Fig.~\ref{fig:vqsvd-demo1}(e), optimizations with a smaller cutoff yield higher precision.
Therefore, this difference-based estimation, anchored by such high-precision calculations, is expected to enhance the accuracy of the dominant singular values.
Crucially, this method offers a fundamental advantage in mitigating measurement complexity.
While optimizing for the entire spectrum requires measuring and sorting all individual singular values, this approach strictly requires evaluating only the cumulative sums up to the cutoff, thereby drastically reducing the measurement overhead.
To distinguish between these two paradigms, we hereafter refer to this difference-based method as \textit{partial} optimization and the aforementioned method targeting the entire spectrum as \textit{full} optimization.

\subsection{Simple and Improved Deflation}
\label{subsec:partial_svd_alg2}

To further surpass the limits of single-circuit expressivity, we propose two alternative deflation-based partial SVD algorithms, which employ iterative quantum circuit optimization guided by projection operators.
Specifically, we formulate the optimization problem using projection operators $\hat{P}_{n-1}$ constructed from the unitary gates optimized in the preceding steps, $\hat{U}_{n-1}$ and $\hat{V}_{n-1}$ (the explicit construction of these operators is detailed below).
The optimization iteration is given by,
\begin{align}
&I_{n}(\hat{U}_{n},\hat{V}_{n})=\mathrm{Re}\langle 0|\hat{U}^{\dagger}_{n}\otimes\hat{V}^{\dagger}_{n}|\Phi_{n-1}\rangle,\\
&|\Phi_{n-1}\rangle=\alpha_{n-1}^{-1}(\hat{I}-\hat{P}_{n-1})|\Phi\rangle,
\end{align}
where $\alpha_{n-1}$ is the normalization factor for the projected residual state $|\Phi_{n-1}\rangle$.
This method approximates $|\Phi_{n-1}\rangle$ with the product state $\hat{U}_{n}\otimes\hat{V}_{n}|0\rangle$.
This is mathematically equivalent to optimizing the quantum circuits to capture the largest singular value of the residual state $|\Phi_{n-1}\rangle$, implicitly aligning the global phase to yield a positive real overlap.
By sequentially determining $\hat{U}_{n}$ and $\hat{V}_{n}$ to maximize $I_{n}(\hat{U}_{n},\hat{V}_{n})$, we systematically optimize the circuits $\hat{U}_{n}$ and $\hat{V}_{n}$ corresponding to the left and right singular vectors.
This methodology shares its sequential extraction philosophy with the Variational Quantum Deflation (VQD) algorithm~\cite{higgott_Variational_2019a} used for excited-state searches, offering a stark contrast to our full optimization method, which employs a simultaneous optimization approach analogous to the Subspace-Search Variational Quantum Eigensolver (SSVQE)~\cite{nakanishi_Subspacesearch_2019}.
Our aim is to enhance the accuracy of individual singular values by iteratively assigning dedicated quantum circuits to each principal component.

To implement this deflation procedure, we consider two distinct constructions of the projection operators.
The first approach, termed \textit{simple} deflation, constructs the projection operator directly from the obtained quantum states $|u_{n}\rangle=\hat{U}_{n}|0\rangle$ and $|v_{n}\rangle=\hat{V}_{n}|0\rangle$ as follows:
\begin{align}
\hat{P}_{n}=&\sum^{n}_{m=1}|u_{m},v_{m}\rangle\langle u_{m},v_{m}|.
\end{align}
Using this operator, the target singular value can be simply defined as $s_{n}=\mathrm{Re}\langle u_{n},v_{n}|\Phi\rangle$.
However, this naive formulation relies on the strict assumption that the vector sets $\{|u_{n}\rangle\}$ and $\{|v_{n}\rangle\}$ are perfectly orthonormal.
In practice, due to the limited expressivity of shallow quantum circuits and inherent hardware noise, this mutual orthogonality is rarely satisfied, leading to severe numerical instabilities.

To fundamentally resolve this orthogonality breakdown, we introduce the \textit{improved} deflation algorithm.
Unlike simple deflation, this method explicitly corrects for the non-orthogonality of the variational states via classical post-processing of the overlap matrices.
It utilizes a refined projection operator defined as:
\begin{align}
\hat{P}_{n}=\hat{P}^{A}_{n}\otimes\hat{P}^{B}_{n}
\end{align}
where the projection operators $\hat{P}^{A}_{n}$ and $\hat{P}^{B}_{n}$ for the respective vector sets $\{|u_{n}\rangle\}$ and $\{|v_{n}\rangle\}$ are constructed as:
\begin{align}
\hat{P}^{A}_{n}=&\sum^{n}_{k,l=1}|u_{k}\rangle
\langle u_{l}|[A^{+}_{n}]_{k,l},\\
\hat{P}^{B}_{n}=&\sum^{n}_{k,l=1}|v_{k}\rangle
\langle v_{l}|[B^{+}_{n}]_{k,l}.
\end{align}
Here, $A^{+}_{n}$ and $B^{+}_{n}$ denote the pseudo-inverses of the overlap matrices: $A_{n}=(\langle u_{k}|u_{l}\rangle)^{n}_{k,l=1}$ and $B_{n}=(\langle v_{k}|v_{l}\rangle)^{n}_{k,l=1}$, respectively.
While $A_{n}$ and $B_{n}$ are Hermitian with non-negative eigenvalues, they inherently become ill-conditioned or singular as the variationally obtained states lose mutual orthogonality.
To ensure numerical stability, these pseudo-inverses are computed via spectral decomposition:
\begin{align}
[A^{+}_{n}]_{k,l}=&
\sum^{n}_{m=1}\tau_{\epsilon}(d^{A}_{m})U^{A}_{k,m}(U^{A}_{l,m})^{*},
\\
[B^{+}_{n}]_{k,l}=&
\sum^{n}_{m=1}\tau_{\epsilon}(d^{B}_{m})U^{B}_{k,m}(U^{B}_{l,m})^{*},
\end{align}
where $d^{A}_{m}$ and $d^{B}_{m}$ are the eigenvalues of $A_{n}$ and $B_{n}$, respectively.
The threshold function is defined as $\tau_{\epsilon}(x)=1/x$ for $x>\epsilon$ and $0$ otherwise.
This threshold $\epsilon$ acts as a numerical regularization, filtering out noise originating from optimization errors and stabilizing the inversion of ill-conditioned overlap matrices (in our numerical demonstrations, we set $\epsilon=10^{-12}$).
Using these corrected projection operators, the resulting singular values are extracted from the singular values of the core matrix $\Sigma^{(n)}$:
\begin{align}
\Sigma^{(n)}_{k,l}=&
\tau_{\epsilon}(d^{A}_{k})^{\frac{1}{2}}
\tau_{\epsilon}(d^{B}_{l})^{\frac{1}{2}}
\nonumber\\&\times
\sum^{n}_{k',l'=1}(U^{A}_{k',k})^{*}(U^{B}_{l',l})^{*}
\langle u_{k'},v_{l'}|\Phi\rangle.
\label{eq:corrected_core_matrix}
\end{align}
It is worth noting that the projected state, which provides a high-fidelity approximation of the target state $|\Phi\rangle$, can be explicitly reconstructed as
\begin{align}
|\Phi\rangle\simeq\hat{P}_{n}|\Phi\rangle=&
\sum_{k}\sum_{l}|\chi^{A}_{k},\chi^{B}_{l}\rangle\Sigma^{(n)}_{k,l},
\end{align}
where $|\chi^{A}_{k}\rangle$ and $|\chi^{B}_{k}\rangle$ are the exactly orthonormalized eigenvectors of the projection operators $\hat{P}^{A}_{n}$ and $\hat{P}^{B}_{n}$, given by
\begin{align}
|\chi^{A}_{k}\rangle=&
\sum^{n}_{k'=1}|u_{k'}\rangle U^{A}_{k,k'}\tau_{\epsilon}(d^{A}_{k'})^{\frac{1}{2}},
\\
|\chi^{B}_{k}\rangle=&
\sum^{n}_{k'=1}|v_{k'}\rangle U^{B}_{k,k'}\tau_{\epsilon}(d^{B}_{k'})^{\frac{1}{2}}.
\end{align}

\subsection{Demonstration for a 1D Heisenberg model}
\label{subsec:partial_svd_demo1}


\begin{figure*}[htbp]
\centering
\includegraphics[width=2.0\columnwidth]{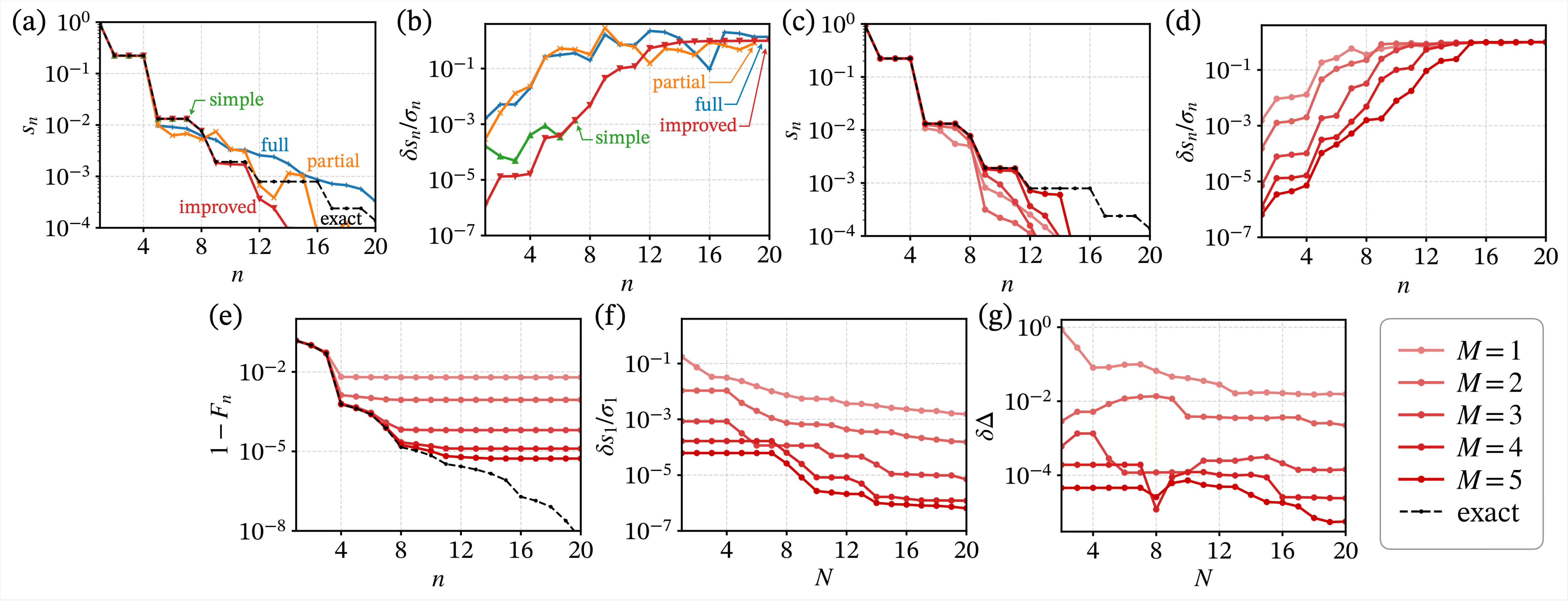}
\caption{
Performance benchmarks of four variational SVD algorithms for the 1D Heisenberg model ($L=16$).
(a) Singular values and (b) their relative errors with the number of layers fixed to $M=4$.
(c)-(g) Detailed results obtained via the improved deflation algorithm, evaluated for varying numbers of layers up to $M=5$.
(c)-(e) Singular values $s_{n}$, relative errors $\delta s_{n}/\sigma_{n}$, and cumulative squared sums $F_{n}$.
(f)(g) Relative error of the largest singular value and absolute error in the Schmidt gap $\delta\Delta$ where $\Delta=\xi_{2}-\xi_{1}$ as a function of the number of deflation steps $N$, shown for various numbers of layers $M$ up to $5$.
}
\label{fig:vqsvd-demo2}
\end{figure*}


To demonstrate the performance of the proposed algorithms, we apply the four variational SVD methods---full optimization, partial optimization, simple deflation, and improved deflation---to the 1D Heisenberg model.
Figures~\ref{fig:vqsvd-demo2}(a) and (b) compare the singular values and their relative errors obtained by each approach, with the circuit depth fixed to $M=4$.
For the improved deflation algorithm, the circuits are sequentially optimized up to $N=20$ steps.
As shown in Figs.~\ref{fig:vqsvd-demo2}(a) and (b), while all methods generally capture the overall singular value distribution, the two deflation-based algorithms achieve markedly higher precision.
Focusing on the relative error $\delta s_{n}/\sigma_{n}$, the partial optimization algorithm accurately reproduces the dominant singular value $s_{1}$, as intended by its formulation. 
However, its accuracy for subdominant singular values degrades rapidly, converging to the performance of the full optimization method.
This behavior exposes a fundamental limitation: even if a cutoff concentrates optimization power on a specific singular value, the total expressivity of a single circuit remains bounded, inherently restricting its capacity to simultaneously resolve multiple orthogonal components.
While the simple deflation algorithm provides a reasonable approximation of the dominant singular values, its precision remains inferior to that of improved deflation.
Furthermore, as anticipated, simple deflation exhibits severe numerical instability for smaller singular values, occasionally outputting unphysical negative values due to the progressive breakdown of mutual orthogonality.
These comparisons demonstrate that the improved deflation algorithm yields the highest accuracy and stability for the dominant singular values among all evaluated methods.
The explicit orthogonality correction via pseudo-inverses proves indispensable for robust entanglement spectrum estimation.

Next, we delve into the detailed performance of the improved deflation algorithm across varying circuit depths $M$, as analyzed in Figs.~\ref{fig:vqsvd-demo2}(c)-(e).
Naturally, deep circuits ($M\ge 4$) possess sufficient expressivity to accurately represent the singular vectors, enabling the successful extraction of numerous singular values [Fig.~\ref{fig:vqsvd-demo2}(c)].
Conversely, shallow circuits inherently lack the expressivity to span the full relevant subspace.
Consequently, increasing the deflation steps $N$ does not continuously yield new effective orthogonal vectors for the subdominant components.
Because the limited expressivity prevents the circuits from reaching these smaller singular values, the relative error for these components fails to decrease [Fig.~\ref{fig:vqsvd-demo2}(d)], and the infidelity $1-F_{n}$ plateaus at a value significantly above $0$ [Fig.~\ref{fig:vqsvd-demo2}(e)].
Crucially, however, despite this saturation, the infidelity $1-F_{n}$ achieved by improved deflation using these shallow circuits remains equivalent to, or strictly lower than, the minimum infidelity obtained by the full optimization algorithm using significantly deeper circuits ($M=4,5$), as previously shown in Fig.~\ref{fig:vqsvd-demo1}(c).

However, this expressivity bound does not render shallow circuits obsolete within our framework.
As illustrated in Fig.~\ref{fig:vqsvd-demo2}(f), increasing the deflation steps $N$ systematically enhances the precision of the already-obtained dominant singular value $s_{1}$, even when using minimally expressive shallow circuits.
This phenomenon highlights the core mechanism of our approach: the classical post-processing acts as an error filtering process.
Initial states extracted by shallow circuits inevitably exhibit spurious overlaps with other orthogonal singular vectors due to their constrained parameter space.
As $N$ increases and subsequent vectors are extracted, the classical orthogonality correction explicitly identifies and mathematically filters out these non-orthogonal components from the dominant subspace, thereby continuously refining the estimate of $s_{1}$.
Thus, the improved deflation algorithm actively compensates for suboptimal quantum circuit expressivity.
Crucially, because the required overlap matrices for such shallow circuits can be efficiently evaluated via classical TN methods, this filtering is achieved entirely without quantum measurement overheads.
For highly expressive deep circuits, the initial residual is already minimal, making this post-processing filtering less prominent but equally robust.

As a direct consequence of this progressive filtering, the structural integrity of the Schmidt decomposition is gradually restored.
Figure~\ref{fig:vqsvd-demo2}(g) demonstrates that the absolute error of the Schmidt gap, $\Delta=\xi_{2}-\xi_{1}$ (where the entanglement spectrum is defined as $\xi_{n}=-2\ln{\sigma_{n}}$), exhibits an overall decreasing trend as $N$ increases.
This confirms that the orthogonality correction successfully resolves the differences between dominant components, enabling a reliable evaluation of the Schmidt gap even under severe circuit constraints.
As demonstrated in Appendix~\ref{app:benchmark_on_1d_spin_ladder}, this capability is highly practical, allowing for the precise detection of entanglement spectrum degeneracies---a hallmark signature of topologically non-trivial phases~\cite{pollmann_Entanglement_2010,miyakoshi_Entanglement_2016}.

\subsection{Demonstration for a 2D Heisenberg model}
\label{subsec:partial_svd_demo2}


\begin{figure*}[htbp]
\centering
\includegraphics[width=2.0\columnwidth]{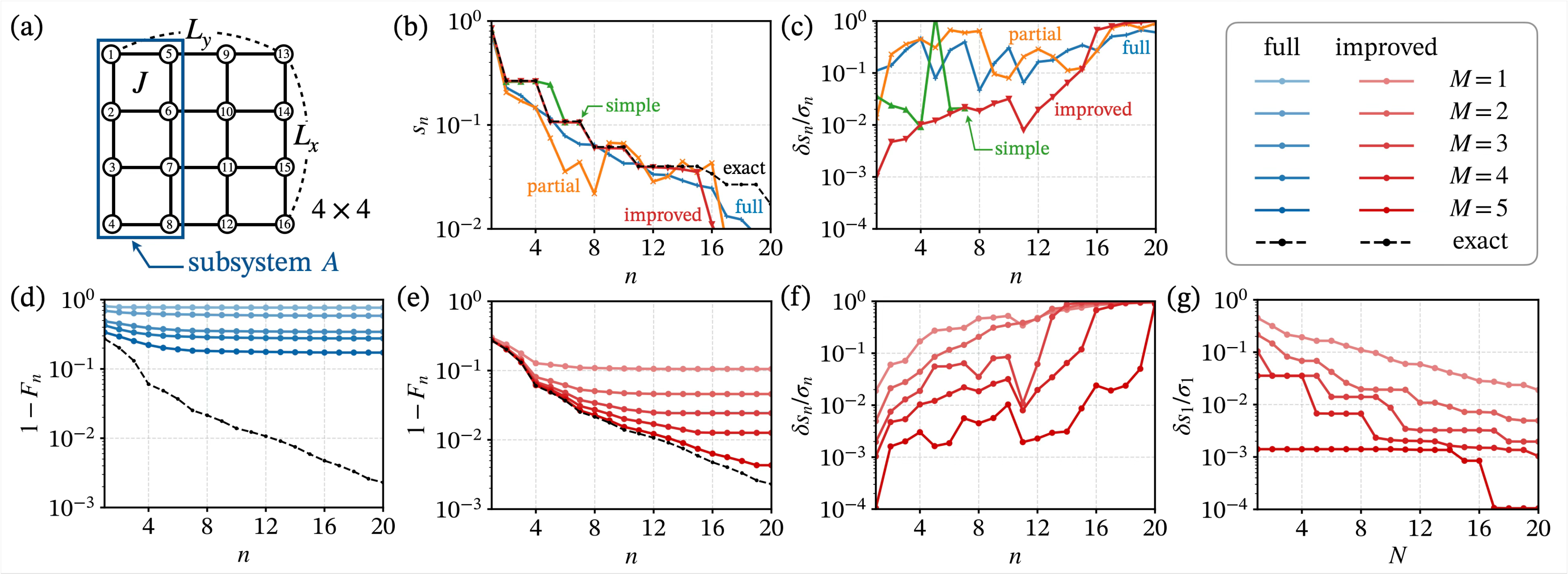}
\caption{
Performance benchmarks of partial SVD algorithms for the 2D Heisenberg model with system size $(L_x,L_y)=(4,4)$.
(a) Schematic representation of the 2D lattice of the Heisenberg model, where each site is labeled by its site index and the rectangular region enclosed by the blue line denotes the subsystem $A$.
(b)(c) Singular value distribution $s_{n}$ and its relative error $\delta s_{n}/\sigma_{n}$, respectively, obtained via the four proposed algorithms with the number of layers  fixed to $M=4$.
(d)-(g) Detailed results evaluated for varying  numbers of layers up to $M=5$.
(d)(e) Comparison of the cumulative squared sums $F_{n}$ between the full optimization (d) and the improved deflation algorithms (e).
(f) Relative error of singular values for the improved deflation algorithm.
(g) Relative error of the largest singular value as a function of deflation steps $N$ for the improved deflation algorithm.
}
\label{fig:vqsvd-demo3}
\end{figure*}


To further validate the scalability and robustness of our approach, we apply the four proposed methods to the Heisenberg model on a 2D square lattice of size $(L_{x},L_{y})=(4,4)$ with open boundary conditions [Fig.~\ref{fig:vqsvd-demo3}(a)].
The system is governed by the Hamiltonian:
\begin{align}
\hat{H}=& J\sum^{L_{x}-1}_{n=1}\sum^{L_{y}}_{m=1}
\hat{\bm{S}}_{n,m}\cdot\hat{\bm{S}}_{n+1,m}
+J\sum^{L_{x}}_{n=1}\sum^{L_{y}-1}_{m=1}
\hat{\bm{S}}_{n,m}\cdot\hat{\bm{S}}_{n,m+1},
\end{align}
where $\hat{\bm{S}}_{n,m}$ is the spin-$1/2$ operator acting on a qubit located at the 2D spatial coordinate $(n,m)$, with $n\in\{1,\cdots,L_{x}\}$ and $m\in\{1,\cdots,L_{y}\}$.
The explicit mapping between these physical 2D coordinates and the 1D qubit site indices of our variational quantum circuits is illustrated in Fig.~\ref{fig:vqsvd-demo3}(a).
In our calculations, we set the uniform exchange coupling to $J=1$.
The bipartite cut defines subsystem $A$ as the rectangular region enclosed by the blue line.

Figures~\ref{fig:vqsvd-demo3}(b) and (c) display the singular value distribution and relative error, respectively, obtained at a fixed number of layers $M=4$.
Consistent with the 1D benchmark, the improved deflation algorithm achieves high precision and robustly maintains numerical stability deep into the subdominant singular values.

Next, we delve into the detailed performance of the algorithms for varying layers up to $M=5$ shown in Figs.~\ref{fig:vqsvd-demo3}(d)-(g), highlighting the fundamental differences between 1D and 2D systems.
A fundamental challenge in 2D systems is the structural mismatch between the 2D entanglement geometry and the aforementioned 1D structure of the variational quantum circuit employed here, which severely restricts the circuit expressivity.
This structural penalty is evident in the full optimization algorithm; as shown in Fig.~\ref{fig:vqsvd-demo3}(d), increasing circuit depth $M$ fails to reduce the infidelity $1-F_{n}$.
In contrast, the improved deflation algorithm circumvents this structural limitation [Fig.~\ref{fig:vqsvd-demo3}(e)].
Consequently, the infidelity $1-F_{n}$ achieved by the improved deflation algorithm with a single-layer circuit ($M=1$) and $N=20$ deflation steps drops below the minimum value obtained by the full optimization algorithm using a deep circuit of $M=5$.
By systematically accumulating multiple 1D quantum circuits through the deflation process, the framework effectively expands the total bond dimension, successfully capturing the 2D entanglement structure despite the restrictive 1D ansatz.

Beyond circuit geometry, the inherent broadness of the 2D entanglement spectrum [Fig.~\ref{fig:vqsvd-demo3}(b)] dictates a fundamentally different optimization landscape compared to the 1D chain.
In the 1D case, the steeply decaying singular values, combined with low circuit expressivity, leave a large residual, forcing the algorithm to perform orthogonality correction to accurately estimate the dominant singular values.
Conversely, in the 2D system, the singular value distribution is characteristically broad.
Faced with this dense spectrum, the objective function inherently prioritizes the exploration and extraction of new, comparatively large singular values over the immediate filtering of the already-obtained ones.
Consequently, as observed in Fig.~\ref{fig:vqsvd-demo3}(f), while the macroscopic broad distribution is successfully captured, the relative error for individual singular values does not drop as rapidly as in the 1D scenario.
This underlying mechanism is further corroborated by Fig.~\ref{fig:vqsvd-demo3}(g); to effectively filter out these non-orthogonal components and refine individual components such as $s_{1}$, the 2D system demands a large number of deflation steps $N$.
Furthermore, this demand for a large $N$ is fundamentally aligned with the physical requirements of TNs.
Since a shallow 1D circuit inherently lacks the bond dimension required for the 2D entanglement structure, increasing the deflation steps $N$ is physically essential to construct a sufficiently expressive linear combination of states.

These findings suggest that accurately resolving higher-dimensional entanglement structures necessitates a synergistic approach: increasing the deflation steps to aggregate effective bond dimension, ideally coupled with tailoring the quantum circuit ansatz to match the target geometry.
Nevertheless, it is a significant result that, even without problem-specific circuit engineering, the improved deflation algorithm robustly recovers the global entanglement properties of 2D systems, affirming its position as a versatile framework.

\section{Discussion on Scalability and Implementation}
\label{sec:discussion}

The numerical results presented in this study demonstrate that the improved deflation algorithm achieves high precision and numerical stability for evaluating the entanglement spectrum.
However, a critical consideration for practical applicability to near-term quantum hardware is the fundamental trade-off between measurement complexity and the severe constraints imposed by circuit depth and hardware noise.
While global optimization methods such as full and partial optimization struggle with this trade-off due to their inherent requirement for deep, noise-susceptible circuits, our deflation-based approaches---particularly the improved deflation algorithm---offer a scalable alternative.
In this section, we detail an efficient hybrid quantum-classical implementation of our algorithm and discuss how this specific architecture inherently mitigates these hardware limitations.


\begin{figure*}[htbp]
\centering
\includegraphics[width=1.6\columnwidth]{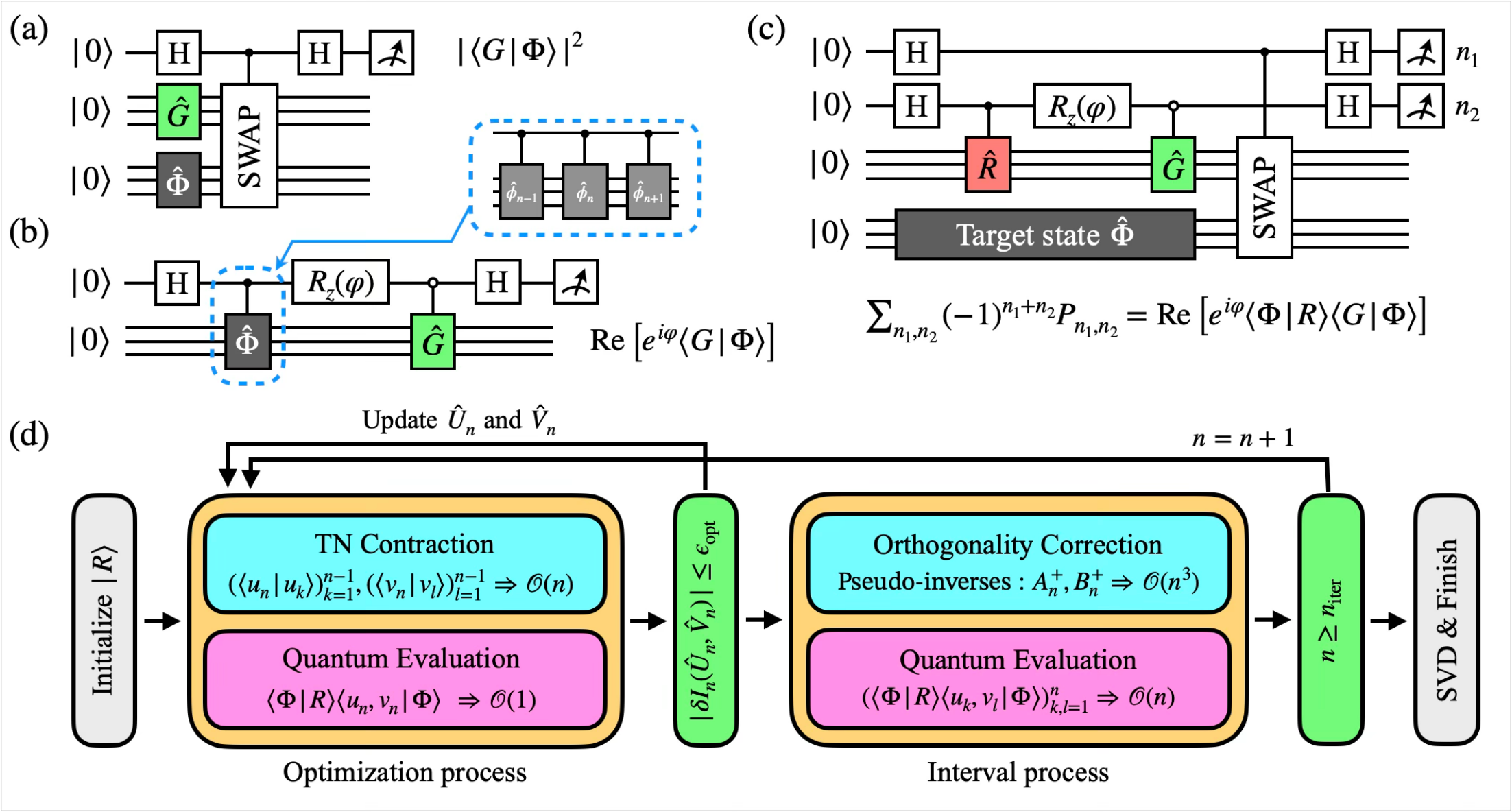}
\caption{
Comparison of hardware implementation strategies and the proposed hybrid architecture for improved deflation.
Here, $|\Phi\rangle=\hat{\Phi}|0\rangle$ represents the complex target state (composed of a sequence of numerous unitary gates: $\hat{\Phi}=\prod_{n}\hat{\phi}_{n}$), $\hat{G}$ denotes the parametrized variational circuit, and $|R\rangle=\hat{R}|0\rangle$ is a shallow reference state.
(a) Quantum circuit of the SWAP test for evaluating the absolute square $|\langle\Phi|G\rangle|^2$.
(b) Quantum circuit of the Hadamard test for evaluating the complex amplitude $\mathrm{Re}\left[e^{i\varphi}\langle G|\Phi\rangle\right]$.
Solid and open circles denote control conditioned on the ancilla states $|1\rangle$ and $|0\rangle$, corresponding to the operations $C_{1}$ and $C_{0}$, respectively.
As indicated by the blue dotted boxes connected by an arrow, naively implementing the controlled target state $C_{1}(\hat{\Phi})$ requires a prohibitive circuit depth, because it must be decomposed into a sequence of numerous controlled gates, $\prod_{n}C_{1}(\hat{\phi}_{n})$.
(c) Proposed quantum circuit diagram for evaluating the complex coefficient $\mathrm{Re}\left[e^{i\varphi}\langle\Phi|R\rangle\langle G|\Phi\rangle\right]$.
Note that this implementation avoids constructing controlled unitary gates for the complex target state $|\Phi\rangle$.
(d) Flowchart of the proposed hybrid quantum-classical algorithm for improved deflation.
The procedure consists of optimization and interval processes, where quantum circuit evaluations and classical tensor network contractions and matrix updates are executed concurrently to maximize computational efficiency and minimize hardware idle time.
}
\label{fig:qsvd-implementation}
\end{figure*}


\subsection{Hybrid Architecture for Improved Deflation}
\label{subsec:hybrid_implementation}

As discussed in Sec.~\ref{subsec:full_svd_alg}, global optimization methods (full and partial optimization) approximate a target set of singular vectors simultaneously by maximizing the absolute value of the overlap between the parametrized circuit state and the complex target state $|\Phi\rangle=\hat{\Phi}|0\rangle$.
Operationally, this cost function is evaluated using a standard SWAP test shown in Fig.~\ref{fig:qsvd-implementation}(a), which completely avoids the need to implement controlled versions of both the target and parametrized circuits.
It is worth noting that this approach exhibits a characteristic scaling in measurement complexity.
If we consider a generic parametrized circuit state $|G\rangle=\hat{G}|0\rangle$, the physically measured signal scales quadratically as $\eta^2$ when the overlap $\eta\sim|\langle G|\Phi\rangle|$ is small.
Resolving this quadratically suppressed signal against statistical shot noise requires a measurement count scaling as $\mathcal{O}(1/\eta^4)$.

In contrast, the deflation-based approach sequentially extracts singular vectors, where the subtraction process requires evaluating the overlap as a complex number to account for relative phase.
By directly evaluating the complex amplitude, one theoretically achieves linear sensitivity (proportional to $\eta$), which drastically reduces the required shot count to $\mathcal{O}(1/\eta^2)$.
A naive approach to evaluate such complex overlaps is the Hadamard test shown in Fig.~\ref{fig:qsvd-implementation}(b).
However, from a practical hardware perspective, this introduces a practical dilemma.
The Hadamard test necessitates implementing controlled unitary gates not only for the parametrized circuit $\hat{G}$ but also for the highly complex target circuit $\hat{\Phi}$.
Constructing the controlled version of the target state preparation requires a prohibitive circuit depth.
Specifically, because the target state is composed of a sequence of numerous unitary gates, $\hat{\Phi}=\prod_{n}\hat{\phi}_{n}$, its controlled version $C_{k}(\hat{\Phi})$---conditioned on an ancilla state $|k\rangle$ with $k\in\{0,1\}$---necessitates controlling every individual operation in sequence:
\begin{align}
C_{k}(\hat{\Phi})=\prod_{n}C_{k}(\hat{\phi}_{n}).
\end{align}
As illustrated by the blue dotted boxes in Fig.~\ref{fig:qsvd-implementation}(b), this multiplicative overhead causes severe noise accumulation on near-term hardware, effectively acting as a strong depolarizing channel with rate $\gamma$.
This introduces an attenuation factor $\lambda=1-\gamma\ll 1$, shrinking the practically measured signal to $\lambda\eta$.
Consequently, the required shot count $\mathcal{O}(1/(\lambda\eta)^2)$ can become prohibitively large in practice, substantially diminishing the theoretical advantage of linear sensitivity.

Therefore, the primary objective of this section is to construct a practical measurement scheme that preserves the linear sensitivity while strictly avoiding the controlled version of the complex target state.
To achieve this, we introduce an interferometric measurement scheme utilizing a shallow reference state $|R\rangle=\hat{R}|0\rangle$, implemented via the circuit in Fig.~\ref{fig:qsvd-implementation}(c).
By evaluating the interference between the target state $|\Phi\rangle$ and an ancilla-controlled superposition of the reference state $|R\rangle$ and the parametrized state $|G\rangle$, we can extract the complex coefficient using only the controlled versions of the parametrized circuit $\hat{G}$ and the reference $\hat{R}$.

To explicitly demonstrate how this complex coefficient is extracted without using the controlled version of the target state, let us trace the quantum state evolution of the interferometric circuit shown in Fig.~\ref{fig:qsvd-implementation}(c), where the first two qubits serve as ancilla qubits and the subsequent qubits correspond to the physical registers. 
The total operation acting on the initial state $|0\rangle^{\otimes 4}$ yields
\begin{widetext}
\begin{align}
&H_{1}H_{2}C_{1}(\mathrm{SWAP})_{1,(34)}C_{0}(\hat{G})_{2,3}R_{z}(\varphi)_{2}C_{1}(\hat{R})_{2,3}H_{2}H_{1}\hat{\Phi}_{4}|0\rangle^{\otimes 4} \nonumber
\\
=& \frac{1}{4}e^{-i\varphi/2}\Biggl[
|0,0\rangle\otimes\left(|G+e^{i\varphi}R,\Phi\rangle+|\Phi,G+e^{i\varphi}R\rangle\right)
+|1,0\rangle\otimes\left(|G+e^{i\varphi}R,\Phi\rangle-|\Phi,G+e^{i\varphi}R\rangle\right) \nonumber \\
&+|0,1\rangle\otimes\left(|G-e^{i\varphi}R,\Phi\rangle+|\Phi,G-e^{i\varphi}R\rangle\right)
+|1,1\rangle\otimes\left(|G-e^{i\varphi}R,\Phi\rangle-|\Phi,G-e^{i\varphi}R\rangle\right)
\Biggr],
\end{align}
\end{widetext}
where we employ the short-hand notation $|G\pm e^{i\varphi}R\rangle\equiv|G\rangle\pm e^{i\varphi}|R\rangle$ to represent the unnormalized linear superposition of state vectors.
Here, $C_{1}(\hat{U})_{i,j}$ denotes a standard controlled unitary gate acting on register $j$ controlled by ancilla qubit $i$, and $C_{0}(\hat{U})_{i,j}=X_{i}C_{1}(\hat{U})_{i,j}X_{i}$ denotes an anti-controlled unitary gate.
Let $P_{n_{1},n_{2}}$ be the measurement probability of obtaining the outcomes $n_{1},n_{2}\in\{0,1\}$ for the ancilla qubits.
These probabilities can be explicitly expressed as
\begin{align}
P_{n_{1},n_{2}}=&\frac{1}{8}\Bigl[
2+2(-1)^{n_{2}}\mathrm{Re}(e^{i\varphi}\langle G|R\rangle)
\nonumber\\
&+(-1)^{n_{1}}|\langle G+(-1)^{n_{2}}e^{-i\varphi}R|\Phi\rangle|^2
\Bigr].
\end{align}
Based on the parity of the measurement outcomes, we can evaluate the following expectation value:
\begin{align}
\sum_{n_{1},n_{2}}(-1)^{n_{1}+n_{2}}P_{n_{1},n_{2}}=&
\frac{1}{4}\sum_{\sigma=\pm 1}
\sigma|\langle G+\sigma e^{-i\varphi}R|\Phi\rangle|^2
\nonumber\\
=&\mathrm{Re}\left[e^{i\varphi}\langle\Phi|R\rangle\langle G|\Phi\rangle\right].
\end{align}
That is, by measuring with the rotation angle of the RZ gate set to $\varphi=0$ and $\varphi=\pi/2$, we can successfully reconstruct the complex coefficient $\langle \Phi|R\rangle\langle G|\Phi\rangle$.

To apply this measurement scheme to our partial SVD algorithm, we substitute the generic parametrized state $|G\rangle$ with the bipartite ansatz $|u_{n},v_{n}\rangle=\hat{U}_{n}\otimes\hat{V}_{n}|0\rangle$.
Using the complex coefficient extracted above, the cost function for the improved deflation can be reformulated as:
\begin{align}
I_{n}(\hat{U}_{n},\hat{V}_{n})=&\mathrm{Re}\Biggl[
\langle\Phi|R\rangle\langle u_{n},v_{n}|\Phi\rangle
-\sum^{n-1}_{k,k'=1}\sum^{n-1}_{l,l'=1}
\nonumber\\&\times
\langle u_{n}|u_{k}\rangle
\langle v_{n}|v_{l}\rangle
[A^{+}_{n-1}]_{k,k'}[B^{+}_{n-1}]_{l,l'}
\nonumber\\&\times
\langle\Phi|R\rangle\langle u_{k'},v_{l'}|\Phi\rangle
\Biggr].
\label{eq:improved_cost_with_reference}
\end{align}
Because the reference state $|R\rangle$ is strictly fixed throughout the algorithm, the inner product $\langle \Phi|R\rangle$ acts as a global constant complex scalar.
Consequently, optimizing this modified cost function perfectly reproduces the parameter gradients and relative phases required for the SVD.
To evaluate the actual singular values after optimization, one simply rescales the extracted complex coefficients $\langle\Phi|R\rangle\langle u_{k},v_{l}|\Phi\rangle$ by dividing them by the magnitude $|\langle\Phi|R\rangle|$, which is easily predetermined via a SWAP test.
Since any remaining global phase $e^{i\mathrm{arg}\langle\Phi|R\rangle}$ does not affect the singular values, this minimal post-processing exactly recovers the true singular value spectrum.

Turning to the practical implementation, a natural concern arises regarding the controlled circuits $C_{k}(\hat{G})$ and $C_{k}(\hat{R})$; deep circuits inevitably induce significant noise accumulation.
Here, the fundamental advantage of our improved deflation algorithm provides a natural solution.
Because our approach implements the deflation process combined with an explicit orthogonality correction, the variational ansatz remains shallow.
Consequently, the controlled version of the parametrized circuit $C_{k}(\hat{G})$ can be compressed to a depth executable on near-term quantum devices.
Crucially, this inherent shallowness also dictates the optimal strategy for evaluating the remaining terms in Eq.~\eqref{eq:improved_cost_with_reference}, namely the overlap matrices $\langle u_{n}|u_{k}\rangle$ and $\langle v_{n}|v_{l}\rangle$.

While these overlaps can be evaluated directly via quantum measurements, our hybrid architecture offloads their computation to classical TN contractions.
This architectural division is not merely a matter of convenience; it ensures that the overlap matrices and their pseudo-inverses $A^{+}_{n}$ and $B^{+}_{n}$ are computed with classical double-precision accuracy.
By protecting the overlap matrix elements from hardware noise, we can construct a robust projection framework.
This rigidly enforced orthogonality prevents the leakage of previously extracted components, which is essential for safely subtracting already-optimized singular vectors without error accumulation.
Furthermore, while the reference state overlap $|\langle \Phi | R \rangle|$ acts as an effective attenuation factor for the linearly sensitive signal, we can secure a non-vanishing baseline (e.g., $|\langle \Phi | R \rangle| \sim \mathcal{O}(1)$) by pre-optimizing a simple low-entanglement ansatz via the SWAP test.
Because $|R\rangle$ only needs to capture the dominant features of the target state rather than its exact correlations, this preliminary step ensures sufficient signal intensity and prevents error amplification when dividing by the overlap $|\langle\Phi|R\rangle|$, while keeping $C_{k}(\hat{R})$ relatively shallow.
Thus, this synergistic design---combining the proposed quantum measurement scheme with the classical TN offloading---effectively leverages the inherent circuit shallowness of improved deflation, providing a practical and noise-resilient framework.

\subsection{Computational Cost and Concurrent Execution}
\label{subsec:resource_analysis}

Having established the hybrid quantum-classical formulation of the improved deflation method, we now analyze its computational cost and scalability at each optimization step.
Figure~\ref{fig:qsvd-implementation}(d) illustrates the flowchart of this hybrid framework.
The overall algorithm consists of two primary tasks---the \textit{optimization} process, which updates the parametrized quantum circuits $\hat{U}_{n}$ and $\hat{V}_{n}$, and the \textit{interval} process, which prepares the classical orthogonality correction.
As shown in the figure, quantum and classical computation tasks are executed concurrently within each process.
In the following, we estimate the computational cost of these processes according to the workflow of our framework and clarify how this concurrent architecture supports practical scalability.

First, we initialize a reference state $|R\rangle$ that has a sufficient overlap with the target state $|\Phi\rangle$.
Subsequently, at the $n$-th deflation step, the optimization process iteratively updates the parametrized quantum circuits $\hat{U}_{n}$ and $\hat{V}_{n}$ according to the cost function formulated in Eq.~\eqref{eq:improved_cost_with_reference}.
During this optimization, the overlap terms $(\langle u_{n}|u_{k}\rangle)^{n-1}_{k=1}$ and $(\langle u_{n}|u_{l}\rangle)^{n-1}_{l=1}$ must be dynamically evaluated.
As discussed in Sec.~\ref{subsec:hybrid_implementation}, to ensure the numerical stability of the projection framework, these overlaps are calculated via classical TN contractions, imposing a classical computational overhead of $\mathcal{O}(n)$ per iteration.
In contrast, the quantum computational task during the optimization process requires the dynamic evaluation of the cross term $\langle\Phi|R\rangle\langle u_{n},v_{n}|\Phi\rangle$, which involves an $\mathcal{O}(1)$ number of quantum evaluations per cost function call (typically two circuit settings with $\varphi=0$ and $\varphi=\pi/2$ for the real and imaginary parts, respectively).
Thus, the quantum measurement overhead in each optimization step remains lower than that of conventional VQE, where measurements typically scale with the number of Hamiltonian terms.

Once the optimization satisfies the convergence criterion (e.g., $|\delta I_{n}(\hat{U}_{n},\hat{V}_{n})|\le\epsilon_\mathrm{opt}$), the algorithm proceeds to the interval process.
At the $n$-th deflation step, this process updates the pseudo-inverses $A^{+}_{n}$ and $B^{+}_{n}$ for the orthogonality correction, and evaluates the cross terms associated with the newly obtained quantum circuits $\hat{U}_{n}$ and $\hat{V}_{n}$: $(\langle\Phi|R\rangle\langle u_{k},v_{l}|\Phi\rangle)^{n}_{k,l=1}$.
Specifically, the quantum evaluations in this step are limited to cross terms involving the most recently extracted circuits---namely, $(\langle\Phi|R\rangle\langle u_{k},v_{n}|\Phi\rangle)^{n-1}_{k=1}$ and $(\langle\Phi|R\rangle\langle u_{n},v_{l}|\Phi\rangle)^{n-1}_{l=1}$---so that the quantum cost is bounded by $\mathcal{O}(n)$.
Meanwhile, the classical computation is dominated by the pseudo-inverses of the overlap matrices, which incur a diagonalization cost of $\mathcal{O}(n^3)$.
This explicit orthogonality correction is a defining feature of the improved deflation method.
For comparison, a similar hybrid implementation can be constructed for the simple deflation algorithm by omitting this interval process.
However, as discussed in Sec.~\ref{subsec:partial_svd_demo1}, simple deflation requires a sufficiently deep and expressive variational ansatz to remain reliable.
As the number of extracted singular vectors $n$ increases, it induces severe ill-conditioning in the overlap matrices, leading to numerical instabilities that hinder the meaningful construction of quantum circuits.
While the pseudo-inverses in the improved deflation introduce an $\mathcal{O}(n^3)$ scaling bottleneck for sufficiently large $n$, this overhead appears only in the interval process and does not require updates at every optimization step; thus, its practical overhead is substantially reduced.
Executing this classical post-processing enhances numerical stability and accuracy while maintaining a relatively manageable computational cost, thereby providing a scalable implementation.

Finally, when the number of deflation steps reaches the iteration limit $n_\mathrm{iter}$, the workflow concludes by constructing the modified core matrix $e^{i\mathrm{arg}\langle\Phi|R\rangle}\Sigma^{(n)}$ (where $\Sigma^{(n)}$ is defined in Eq.~\eqref{eq:corrected_core_matrix}) using the accumulated cross terms.
Then, the final singular values can be extracted from the SVD of this matrix, which also provides the coefficients required to construct the final singular vectors.

Crucially, because the classical TN processing via pseudo-inverses strictly enforces mutual orthogonality, it constructs a well-defined projection operator even from suboptimally extracted states, effectively filtering out the non-orthogonality errors within the subspace.
Consequently, the residual errors propagating to the final singular values are essentially isolated to the hardware noise affecting the quantum evaluation of the cross terms.
This error isolation highlights the complementary structure of our hybrid framework.
High-precision orthogonality corrections, which are inherently challenging for near-term quantum devices, are entirely offloaded to the classical TN processing, while the quantum processor is dedicated to the classically intractable evaluation of cross terms involving complex target states, such as those generated by long-time Trotterized evolutions.

Furthermore, this strategic task allocation yields architectural synergies beyond the aforementioned algorithmic error resilience.
Ensuring mutual orthogonality through classical post-processing enhances the effective expressivity of the extracted-state superposition, allowing the individual quantum circuits in the improved deflation method to remain shallow.
Maintaining these shallow circuits not only mitigates the hardware noise during the quantum cross-term evaluations but also reduces the computational cost of the classical TN contractions.
From an architectural perspective, executing these lightweight quantum and classical tasks concurrently helps minimize the idle time of classical processors during quantum measurements---a persistent bottleneck in conventional variational quantum algorithm (VQA) frameworks.
By clarifying this quantum-classical division of labor, the architecture optimizes overall computational throughput and provides a practical and scalable foundation for tightly integrated high-performance computing and quantum-processing environments.

\section{Conclusion}
\label{sec:conclusion}

In this study, we developed a hybrid quantum-classical variational framework for extracting the entanglement spectrum of quantum many-body systems.
To provide a clear theoretical foundation, we established a correspondence between the canonical form of matrix product states (MPS) and bipartite quantum circuits, yielding a general and intuitive formulation for our variational ansatz.
Based on this architecture, our primary objective was to accurately evaluate the dominant and subdominant Schmidt components of target states prepared by quantum circuits, such as those generated via Trotterized time evolution, which are increasingly challenging to access using classical tensor network (TN) simulations.
As a proof of principle, we benchmarked our framework on the ground states of 1D and 2D Heisenberg models, demonstrating that our improved deflation algorithm accurately captures these essential entanglement properties.
Crucially, by approximating individual singular vectors with multiple quantum circuits and explicitly correcting their non-orthogonality via classical post-processing, this approach achieves high precision and enhanced numerical stability, overcoming the inherent limitations of other partial SVD formulations lacking explicit orthogonality correction.

The significance of these results lies in introducing a novel framework for near-term quantum algorithms: decoupling the burden of accuracy from the quantum circuit optimization process.
Conventional variational frameworks (e.g., targeting entire singular vectors in our full optimization) often face a severe scalability bottleneck, as achieving high numerical accuracy requires deep, highly expressive circuits that inevitably suffer from barren plateaus and hardware noise.
The improved deflation method addresses this bottleneck by utilizing the explicit orthogonality correction as an error-filtering mechanism.
By dynamically filtering out the non-orthogonality errors caused by the limited expressivity of shallow circuits, our approach tolerates suboptimal local minima and hardware noise, and relinquishes the strict requirement for deep or perfectly optimized quantum circuits.

Furthermore, this tolerance for shallow circuits unlocks deep architectural synergies that directly mitigate the fundamental bottlenecks of hardware implementation.
By incorporating an auxiliary shallow reference state, which provides a sufficient approximation to the target state, our framework bypasses the prohibitive overhead of implementing a controlled operation of the complex target state, thereby preserving linear signal sensitivity while suppressing noise.
Simultaneously, the inherent shallowness of the optimized circuits ensures that evaluating the required overlap matrices via classical TN contractions remains computationally efficient and highly accurate.
This strategic task allocation enables a concurrent execution model: the quantum processor focuses exclusively on the classically intractable evaluation of cross terms involving the target state (e.g., $\langle u_{n},v_{m}|\Phi\rangle$), while the classical processor handles the orthogonality corrections and TN contractions.
By minimizing hardware idle time and maximizing overall computational throughput, this synergistic architecture provides a complementary route to modern randomized techniques such as Classical Shadow Tomography (CST), specifically bypassing the exponential sampling and computational overheads typically required when reconstructing the entanglement spectrum of large subsystems~\cite{huang_Predicting_2020}.

Consequently, by systematically mitigating the fundamental hurdles of circuit depth, optimization hardness, and measurement complexity, our framework provides a robust and practical roadmap for large-scale entanglement spectrum estimation.
It offers a relevant pathway for advanced NISQ and post-NISQ devices equipped with error mitigation, while establishing a solid foundation for early fault-tolerant quantum computing applications.
A promising direction for future research involves integrating this approach with sample-efficient measurement protocols, such as CST, and quantum error mitigation techniques, alongside the development of noise-resilient quantum circuit optimization.
Because our framework structurally isolates the quantum computational burden strictly to the evaluation of cross terms, these future optimizations can further reduce the overall implementation cost without altering the core algorithm.
We anticipate that this hybrid framework will provide a useful route for quantitatively characterizing exotic phenomena---such as quantum criticality, topological order, and non-equilibrium dynamics---in complex quantum many-body systems.


\begin{acknowledgments}
This work was supported by Grant-in-Aid for Scientific Research (A) (No.~JP21H04446), Grant-in-Aid for Scientific Research (B) (No.~JP21H03455, No.~JP22H01171, and No~JP24K02948), Grant-in-Aid for Scientific Research (C) (No.~JP22K03479), and Grant-in-Aid for Early-Career Scientists (No.~JP24K16978) from MEXT, Japan, and Grant-in-Aid for Transformative Research Areas "The Natural Laws of Extreme Universe---A New Paradigm for Spacetime and Matter from Quantum Information" (No.~JP21H05182 and No.~JP21H05191) from JSPS, Japan.
It was also supported by JST PRESTO (No.~JPMJPR1911), MEXT Q-LEAP (No.~JPMXS0120319794), and JST COI-NEXT (No.~JPMJPF2014 and No.~JPMJPF2221). 
It was also partially supported by Program for Promoting Research on the Supercomputer Fugaku (No.~JPMXP1020230411) form MEXT, Japan, by the New Energy and Industrial Technology Development Organization (NEDO) (No. JPNP20017), and by the COE research grant in computational science from Hyogo Prefecture and Kobe City through Foundation for Computational Science.
Additionally, we acknowledge the support from the UTokyo Quantum Initiative and the RIKEN TRIP initiative (RIKEN Quantum).
We are grateful for allocating computational resources of the HOKUSAI BigWaterfall supercomputing system at RIKEN and SQUID at the D3 Center, Osaka University.
\end{acknowledgments}



\appendix

\section{Quantum Circuit Representation of the Canonical MPS}
\label{app:quantum_circuit_representation_for_mps}


\begin{figure*}[htbp]
\centering
\includegraphics[width=1.7\columnwidth]{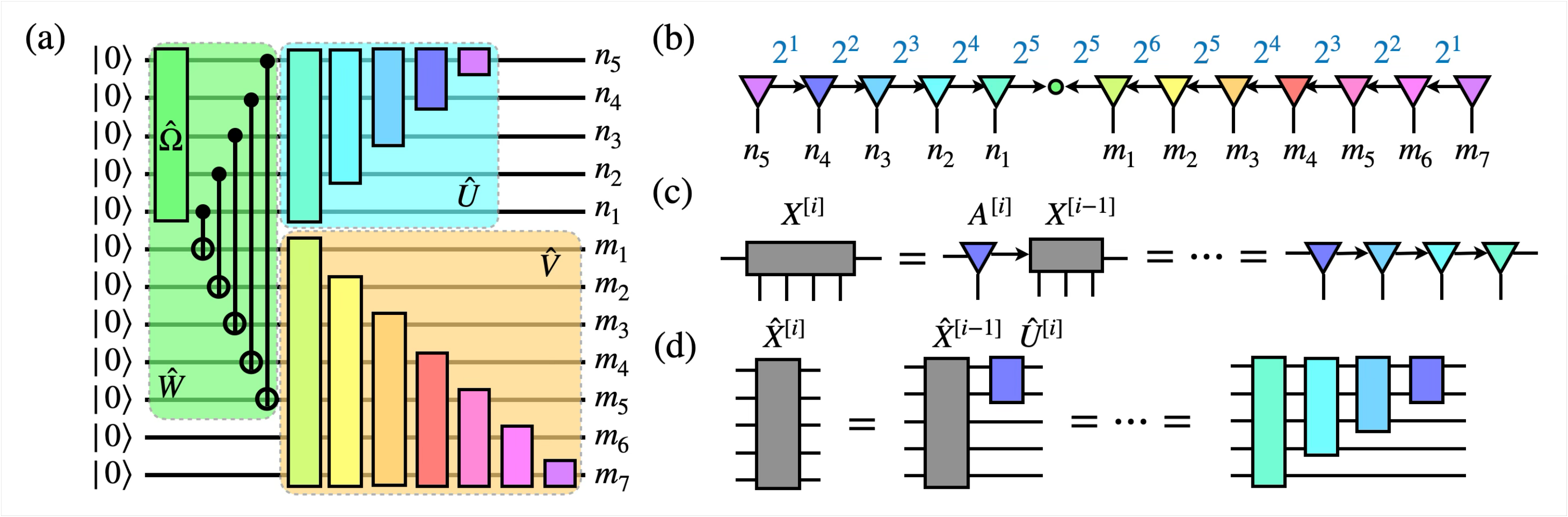}
\caption{
(a) Exact quantum circuit representation of an arbitrary quantum state for a system size $(N,M)=(5,7)$, which corresponds to the canonical form of the MPS with an orthogonalized center shown in (b).
The subscripts $n_{i}$, $m_{j}$ denote the physical indices of each qubit, and the blue numbers on the bond between MPS tensors indicate the bond dimensions.
(c) The iterative process for decomposing large multi-qubit tensors $X^{[i]}$, $Y^{[j]}$ into the MPS tensors, and their quantum circuit representation depicted in (d).
}
\label{fig:svd-gate-decomposition}
\end{figure*}


\subsection{Exact Quantum Circuit Representation}
\label{app:exact_quantum_circuit_representation}

In this section, we demonstrate that a Matrix Product State (MPS) in the canonical form can be exactly mapped to a quantum circuit representation.
To establish this property, we first show that an arbitrary quantum state $|\Psi\rangle$ can be directly transformed into an MPS through Singular Value Decomposition (SVD).
Furthermore, we clarify that this MPS, constructed in the canonical form with an orthogonalized center, directly corresponds to the quantum circuit representation.

First, let us consider an arbitrary quantum state $|\Psi\rangle$ partitioned into two subsystems $A$ and $B$:
\begin{align}
|\Psi\rangle = \sum^{2^{N}-1}_{n=0}
\sum^{2^{M}-1}_{m=0}C_{n,m}|n,m\rangle,
\end{align}
where $n$ and $m$ correspond to the bit-strings $n=n_{N}\cdots n_{1}$ and $m=m_{M}\cdots m_{1}$ respectively, and $|n,m\rangle=|n\rangle\otimes |m\rangle$.
We apply SVD to the coefficient matrix $C$ as $C=X\Sigma Y^{T}$, where $X$ and $Y$ are unitary matrices of size $2^{N}\times 2^{N}$ and $2^{M}\times 2^{M}$, respectively.
$\Sigma$ is a $2^{N}\times 2^{M}$ rectangular matrix with $k$-th diagonal element $\Sigma_{k,k}=\sigma_{k}$ and zero off-diagonal elements.
The state $|\Psi\rangle$ can then be expressed in the form of a Schmidt decomposition:
\begin{align}
&|\Psi\rangle=\sum^{\chi-1}_{k=0}\sigma_{k}\left(\sum^{2^{N}-1}_{n=0}|n\rangle X_{n,k}\right)\otimes\left(\sum^{2^{M}-1}_{m=0}|m\rangle Y_{m,k}\right).
\end{align}
Here, $\chi$ denotes the Schmidt rank of the matrix $C$, satisfying $\chi\le 2^{K}$, where $K=\min(N,M)$.

Next, we sequentially decompose the tensors $X$ and $Y$ into their MPS forms using SVDs. This process corresponds to separating the system one qubit at a time.
To demonstrate this for the tensor $X$, we begin by defining an initial tensor $X^{[N]}_{x_{N}n,k}=X_{n,k}\delta_{x_{N},0}$ with an auxiliary index $x_{N}$ fixed to 0 as:
\begin{align}
X^{[N]}_{x_{N}n_{N}\cdots n_{1},k}=X_{n,k}\delta_{x_{N},0}.
\end{align}
The MPS tensors are generated iteratively, from $i=N$ down to 2.
At each step, we first reshape the tensor $X^{[i]}$ into a matrix by partitioning its indices into a row index $(x_{i},n_{i})$  and a column index $(n_{i-1},\cdots,n_{1},k)$.
We then apply SVD to this matrix, which yields the unitary matrices $U^{[i]}$ and $\mathcal{X}^{[i]}$, and singular values $p^{[i]}_{x_{i-1}}$ as follows:
\begin{align}
&X^{[i]}_{x_{i}n_{i}\cdots n_{1},k}=\sum_{x_{i-1}}p^{[i]}_{x_{i-1}}U^{[i]}_{x_{i}n_{i},x_{i-1}}\mathcal{X}^{[i]}_{kn_{i-1}\cdots n_{1},x_{i-1}}.
\end{align}
Subsequently, we identify the left unitary matrix $U^{[i]}$ as the $i$-th MPS tensor and absorb the remaining singular values $p^{[i]}$ and the right unitary matrix $\mathcal{X}^{[i]}$ into the subsequent tensor $X^{[i-1]}$.
Through these steps, the relations between the tensors $X^{[i]}$ can be summarized as the following recurrence formula:
\begin{align}
X^{[i]}_{x_i n_i \cdots n_1, k} = \sum_{x_{i-1}} U^{[i]}_{x_i n_i, x_{i-1}} X^{[i-1]}_{x_{i-1}, n_{i-1}\cdots n_1, k}
\end{align}
Repeating the above operations down to $i=2$, we can decompose the original tensor $X$ into the product of unitary matrices $U^{[N]},\cdots,U^{[2]}$ and a remaining tensor $X^{[1]}$.
Similarly, the tensor $Y$ can be decomposed into a product of unitary matrices $V^{[M]},\cdots,V^{[1]}$ and a remaining tensor $Y^{[1]}$.
Then, the dimensions of the bond indices $\chi^{X}_{i}$ and $\chi^{Y}_{i}$ satisfy the following inequalities:
\begin{align}
&\chi^{X}_{N}=1,~
\chi^{X}_{i-1}\le\min(2\chi^{X}_{i},2^{i-1}\chi),\\
&\chi^{Y}_{M}=1,~
\chi^{Y}_{i-1}\le\min(2\chi^{Y}_{i},2^{i-1}\chi).
\end{align}
After these sequential decompositions, we apply the SVD at the center of the two subsystems as
\begin{align}
\sum^{\chi-1}_{k=0}\sigma_{k}X^{[1]}_{k,n_{1},x_{1}}
Y^{[1]}_{k,m_{1},y_{1}}=&
\sum^{\chi_{0}-1}_{k=0}w_{k}U^{[1]}_{k,x_{1}n_{1}}
V^{[1]}_{k,y_{1}m_{1}}.
\end{align}
This process yields the new unitary matrices $U^{[1]}$, $V^{[1]}$ and singular values $w_{k}$.
If we do not truncate the bond dimension throughout the process, the singular values $w_{k}$ and the rank $\chi_{0}$ are identical to $\sigma_{k}$ and $\chi$, due to the unitarity of the matrices $X^{[1]}$ and $Y^{[1]}$.
Using the above results, an arbitrary quantum state $|\Psi\rangle$ can be exactly represented as
\begin{align}
|\Psi\rangle=&
\sum^{\chi_{0}-1}_{k=0}
\sum^{2^{N}-1}_{n=0}
\sum^{2^{M}-1}_{m=0}
w_{k}|n,m\rangle
\nonumber\\&\times
\left(\sum_{\{x_{i}\}}U^{[N]}_{x_{N}n_{N},x_{N-1}}
\cdots U^{[1]}_{x_{1}n_{1},k}\right)
\nonumber\\&\times
\left(\sum_{\{y_{i}\}}V^{[M]}_{y_{M}m_{M},y_{M-1}}
\cdots V^{[1]}_{y_{1}m_{1},k}\right).
\end{align}

Next, to rewrite the state $|\Psi\rangle$ into the MPS form, we define the tensors as follows:
\begin{align}
&W_{x',x} = w_{x}\delta_{x',x},
\label{eq:mps-qc-transform-center}
\\
&\Bigl(A^{[n_{i}]}_{i}\Bigr)_{x_{i},x_{i-1}}
=U^{[i]}_{x_{i}n_{i},x_{i-1}},
\label{eq:mps-qc-transform-left}
\\
&\Bigl(B^{[m_{i}]}_{i}\Bigr)_{y_{i-1},y_{i}}
=V^{[i]}_{y_{i}m_{i},y_{i-1}},
\label{eq:mps-qc-transform-right}
\end{align}
where $A^{[n_{i}]}_{i}$ and $B^{[m_{i}]}_{i}$ are the tensors defined for each qubit $|n_{i}\rangle$ and $|m_{i}\rangle$, and $W$ is the diagonal matrix with the singular values $w_{i}$.
Hence, we can rewrite the quantum state $|\Psi\rangle$ into the compact MPS form using the matrix products and taking the $(1,1)$ element:
\begin{align}
|\Psi\rangle=&\sum^{2^{N}-1}_{n=0}\sum^{2^{M}-1}_{m=0}
|n,m\rangle\left(A^{[n_{N}]}_{N}\cdots A^{[n_{1}]}_{1}WB^{[m_{1}]}_{1}\cdots B^{[m_{M}]}_{M}\right)_{1,1}.
\end{align}
It is worth noting that the tensors $A$ and $B$ satisfy the following left and right canonical conditions due to the unitarity of the matrices $U$ and $V$. 
\begin{align}
\sum_{n_{i}}\sum_{x_{i}}
\Bigl(A^{[n_{i}]}_{i}\Bigr)_{x_{i},x_{i-1}}
\Bigl(A^{[n_{i}]}_{i}\Bigr)^{*}_{x_{i},x'_{i-1}}
=\delta_{x_{i-1},x'_{i-1}},
\\
\sum_{m_{i}}\sum_{y_{i}}
\Bigl(B^{[m_{i}]}_{i}\Bigr)_{y_{i-1},y_{i}}
\Bigl(B^{[m_{i}]}_{i}\Bigr)^{*}_{y'_{i-1},y_{i}}
=\delta_{y_{i-1},y'_{i-1}}.
\end{align}
Thus, we can exactly represent an arbitrary quantum state $|\Psi\rangle$ in the canonical form of MPS with the orthogonalized center.

Finally, we show that the above state can be rewritten as the quantum circuit representation.
For this purpose, we introduce the unitary gates $\hat{U}^{[l]}$ and $\hat{V}^{[l]}$ defined by:
\begin{align}
&\hat{U}^{[l]}=\sum_{n_{l}=0,1}|x_{l}n_{l}\rangle
\langle x_{l-1}|U^{[l]}_{x_{l-1},x_{l}n_{l}},\\
&\hat{V}^{[l]}=\sum_{m_{l}=0,1}|y_{l}m_{l}\rangle
\langle y_{l-1}|V^{[l]}_{y_{l-1},y_{l}m_{l}}.
\end{align}
Furthermore, we define the unitary gate $\hat{W}$ such that it prepares the entangled state from the initial state:
\begin{align}
\hat{W}|0,0\rangle=&
\sum^{2^{K}-1}_{k=0}w_{k}|k,k\rangle
\end{align}
By combining the above three types of unitary gates, the quantum state $|\Psi\rangle$ can be represented as
\begin{align}
|\Psi\rangle =&
\left(\hat{U}^{[N]}\cdots\hat{U}^{[1]}\otimes
\hat{V}^{[M]}\cdots\hat{V}^{[1]}\right)\hat{W}|0\rangle.
\end{align}
The gate sizes of the unitary gates $U^{[i]}$ and $V^{[i]}$, denoted as $d^{U}_{i}$ and $d^{V}_{i}$, satisfy the following relation:
\begin{align}
d^{a}_{i}=&\left\lceil\log_{2}\max(\chi^{a}_{i-1},2\chi^{a}_{i})\right\rceil,
\end{align}
where $a=U,V$ and $\chi^{U}_{0}=\chi^{V}_{0}=\chi_{0}$.
Thus, we find that an arbitrary quantum state can be exactly represented as an MPS with an orthogonalized center without any constraint on the bond dimension, and simultaneously, such an MPS can be transformed into the quantum circuit representation.
These results suggest that the quantum state represented as the MPS with the constraint of the bond dimension can always be represented as the quantum circuit via Eqs.~\eqref{eq:mps-qc-transform-center},\eqref{eq:mps-qc-transform-left},\eqref{eq:mps-qc-transform-right}.
We note that the number of qubits shared by two adjacent unitary gates corresponds to the bond dimension (See Figs.\ref{fig:svd-gate-decomposition}(a) and (b)).

\subsection{Quantum Circuit Representation for the Orthogonalized Center}
\label{app:quantum_circuit_representation_for_orthogonal_center}


\begin{figure*}[htbp]
\centering
\includegraphics[width=1.4\columnwidth]{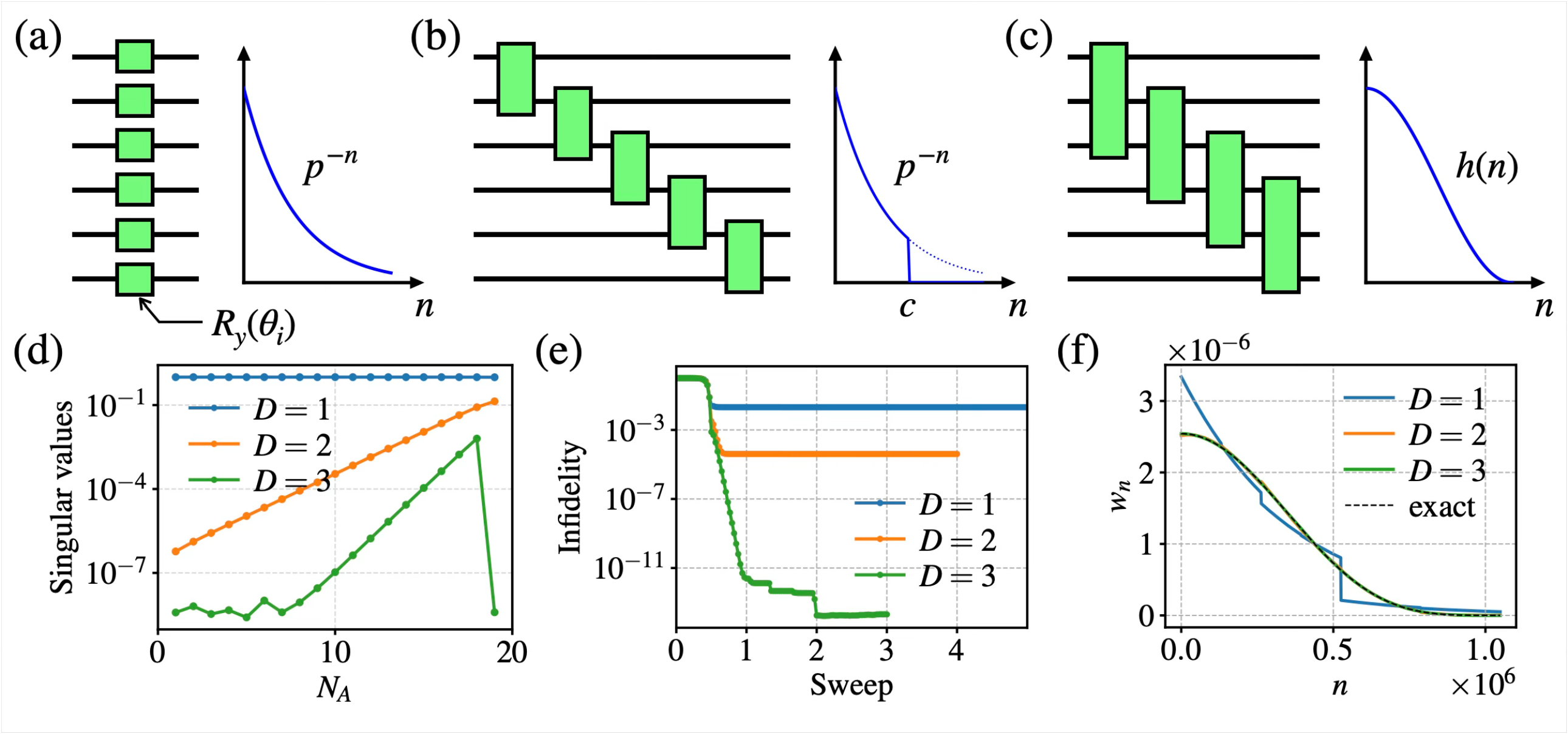}
\caption{
Quantum circuit representation of the orthogonalized center $\hat{\Omega}$ 
for singular values following typical distributions: (a) exponential decay $p^{n}$, (b) exponential decay $p^{n}$ with a cutoff ($w_{n}=0$ for $n\ge c$), and (c) the Hann function.
In panels (a)--(c), the left figures show the quantum circuit structures that can exactly represent the corresponding target distribution profiles displayed in the right figures.
(d) Singular values of the matrix $\tilde{C}$ for the Hann function ($N=20$) plotted as a function of the subsystem boundary $N_{A}$.
(e)(f) Numerical results for the Hann function ($N=20$) using approximated quantum circuits with gate sizes $M=1,2,3$ corresponding to the structures shown in the left panels of (a)--(c): (e) infidelity as a function of the number of sweep updates, and (f) the distribution $w_{n}$ generated by the approximated circuits.
}
\label{fig:qc-svd-examples}
\end{figure*}


In the previous section, we discussed the construction of the quantum circuit for an arbitrary quantum state using SVD, but did not detail the specific quantum circuit representation for each unitary gate.
In this section, we focus on the unitary gates corresponding to the orthogonalized center and construct their detailed quantum circuits for three cases: (1) where the singular values decay exponentially, (2) where they are truncated with a cutoff, and (3) where they are composed of a summation of exponential functions.

First, it is worth noting that the unitary gate $\hat{W}$ corresponding to the orthogonalized center can be decomposed into a $K$-qubit unitary gate $\hat{\Omega}$ with $K=\mathrm{min}(N,M)$ and a product of CNOT gates.
Let $\mathrm{CX}_{i,j}$ denote a CNOT gate with control qubit $i$ and target qubit $j$.
Then, $\hat{W}$ is expressed as:
\begin{align}
&\hat{W}=\prod^{K}_{n=1}\mathrm{CX}_{n,n+N}
(\hat{\Omega}\otimes\hat{I}),\\
&\hat{\Omega}|0\rangle=\sum_{n}w_{n}|n\rangle,
\end{align}
where $\hat{I}$ acts on the remaining qubits, and the coefficients $w_{n}$ satisfy the normalization condition $\sum_{n}|w_{n}|^2=1.$
This decomposition is easily verified from the property of the CNOT gates:
\begin{align}
\prod^{K}_{n=1}\mathrm{CX}_{n,n+N}|n,0\rangle=|n,n\rangle.
\end{align}
Thus, the problem of constructing the quantum circuit for the orthogonalized center $\hat{W}$ reduces to constructing the unitary gate that prepares a quantum state with the weights $w_{n}$ following the singular value distributions.
In the following subsections, we concretely construct the quantum circuits corresponding to several examples of the singular value distributions.

\subsubsection{Exponentially Decaying Singular Values}
\label{app:orthogonal_center_for_exp_function}

We first consider the case where the singular values $w_{n}$ decay exponentially as $w_{n}=pw_{n-1}$ with $0\le p\le 1$.
The distribution of these singular values is plotted in the right panel of Fig.~\ref{fig:qc-svd-examples}(a).
Since the weight factorizes according to the binary representation $n=\sum^{N}_{i=1}2^{i-1}n_{i}$ as $p^{n}=\prod^{N}_{i=1}(p^{2^{i-1}})^{n_{i}}$, the state $\hat{\Omega}|0\rangle$ decomposes into a tensor product of single-qubit states.
Consequently, $\hat{\Omega}$ is implemented as a product of single-qubit rotation gates $R_{y}$ acting on the $i$-th qubit:
\begin{align}
\hat{\Omega}|0\rangle=&
\sqrt{\frac{1}{\mathcal{N}}}
\bigotimes^{N}_{i=1}\sum_{n_{i}\in\{0,1\}}
p^{2^{i-1}n_{i}}|n_{i}\rangle
\nonumber\\=&
\bigotimes^{N}_{i=1}R_{y}(\theta_{i})|0\rangle,
\end{align}
where $\mathcal{N}$ is the normalization factor.
Here, the action of the $R_{y}$ gate is defined by
\begin{align}
R_{y}(\theta_{i})|0\rangle=
\cos{\frac{\theta_{i}}{2}}
\sum_{n_{i}\in\{0,1\}}\left(\tan{\frac{\theta_{i}}{2}}
\right)^{n_{i}}|n_{i}\rangle.
\end{align}
Comparing the coefficients, we find that the normalization factor $\mathcal{N}$, decay ratio $p$, and angle $\theta_{i}$ satisfy the following relations:
\begin{align}
&\mathcal{N}=\prod^{N}_{i=1}
\left(\cos{\frac{\theta_{i}}{2}}\right)^{-2},\\
&p^{2^{i-1}}=\tan{\frac{\theta_{i}}{2}}.
\end{align}
This result demonstrates that for exponentially decaying singular values, the unitary gate $\hat{\Omega}$ can be efficiently represented as a tensor product of local $R_{y}$ gates.

\subsubsection{Exponentially Decaying Singular Values with a Cutoff}
\label{app:orthogonal_center_for_exp_function_with_cutoff}

Next, we consider the construction of the quantum circuit $\hat{\Omega}$ corresponding to the exponentially decaying singular values $w_{n}$ with a cutoff $c$ such that $w_{n}=0$ for $n\ge c$.
The profile of this distribution is illustrated in the right panel of Fig.~\ref{fig:qc-svd-examples}(b).
As demonstrated in this section, the quantum circuit $\hat{\Omega}$ can be composed of a product of unitary gates of gate size 2.
To prove this, it is sufficient to show that the quantum circuit state $\hat{\Omega}|0\rangle$ can be written as an MPS with a bond dimension up to 2.
We first introduce a bipartition $(A,B)$ defined by the cut at an arbitrary position $N_{A}$, where $A=\{1,\cdots,N_{A}\}$ and $B=\{N_{A}+1,\cdots,N_{A}+N_{B}\}$ and assume that the integer $c$ is represented as $c=2^{N_{A}}b+a$.
By decomposing the index $n$ as $n=l+2^{N_{A}}m$, where $l$ ($0\le l\le 2^{N_{A}}$) and $m$ correspond to the lower and higher bits, the state $\hat{\Omega}|0\rangle$ can be written as
\begin{align}
\hat{\Omega}|0\rangle=&\sum^{c-1}_{n=0}w_{n}|n\rangle
=\sqrt{\frac{1}{\mathcal{N}}}\sum^{c-1}_{n=0}p^{n}|n\rangle
\nonumber\\=&
\sqrt{\frac{1}{\mathcal{N}}}\sum^{2^{N_{A}}-1}_{l=0}p^{l}|l\rangle
\otimes\sum^{b-1}_{m=0}p^{2^{N_{A}}m}|m\rangle
\nonumber\\&
+\sqrt{\frac{1}{\mathcal{N}}}\sum^{a-1}_{l=0}p^{l}|l\rangle
\otimes p^{2^{N_{A}}b}|b\rangle.
\end{align}
This result suggests that the maximum bond dimension is fixed to 2 since the state $\hat{\Omega}|0\rangle$ has the mutually orthonormalized bases for subsystems $A$ and $B$.
To confirm this fact, we first choose the orthonormalized bases $\{\varphi^{A}_{n}\}$ and $\{\varphi^{B}_{n}\}$ for the subsystems $A$ and $B$ as follows:
\begin{align}
|\varphi^{A}_{1}\rangle=&\sqrt{\frac{1}{\mathcal{N}^{A}_{1}}}
\sum^{a-1}_{n=0}p^{n}|n\rangle,\\
|\varphi^{A}_{2}\rangle=&\sqrt{\frac{1}{\mathcal{N}^{A}_{2}}}
\sum^{2^{N_{A}}-1}_{n=a}p^{n}|n\rangle,\\
|\varphi^{B}_{1}\rangle=&\sqrt{\frac{1}{\mathcal{N}^{B}_{1}}}
p^{2^{N_{A}}b}|b\rangle,\\
|\varphi^{B}_{2}\rangle=&\sqrt{\frac{1}{\mathcal{N}^{B}_{2}}}
\sum^{b-1}_{n=0}p^{2^{N_{A}}n}|n\rangle,
\end{align}
where $\mathcal{N}^{A}_{i}$ and $\mathcal{N}^{B}_{i}$ are the normalization factors for the bases $\{\varphi^{A}_{n}\}$ and $\{\varphi^{B}_{n}\}$, respectively.
Thus, the state $\hat{\Omega}|0\rangle$ can be written as
\begin{align}
\hat{\Omega}|0\rangle=&
\sum_{n,m}\tilde{C}_{n,m}|\varphi^{A}_{n}\rangle
\otimes|\varphi^{B}_{m}\rangle
\end{align}
where the matrix $\tilde{C}$ is written as
\begin{align}
\tilde{C}=&
\begin{pmatrix}
\sqrt{\mathcal{N}^{A}_{1}\mathcal{N}^{B}_{1}/\mathcal{N}}&
\sqrt{\mathcal{N}^{A}_{1}\mathcal{N}^{B}_{2}/\mathcal{N}}\\
0&\sqrt{\mathcal{N}^{A}_{2}\mathcal{N}^{B}_{2}/\mathcal{N}}
\end{pmatrix}
\nonumber\\\simeq&
\begin{pmatrix}
p^{2^{N_{A}}b}\sqrt{1-p^{2a}}&\sqrt{1-p^{2a}}\\0&p^{a}
\end{pmatrix}
\end{align}
where we assume that the integers $N$ and $N_{A}$ are sufficiently large in the rightmost equation.
This result suggests that the rank of $\tilde{C}$ is at most 2. Thus, the quantum circuit $\hat{\Omega}$ can be written as the product of the sequentially aligned two-qubit gates shown in Fig.~\ref{fig:qc-svd-examples}.
On the other hand, if the matrix $\tilde{C}$ is exactly singular, each of normalization factors $\mathcal{N}^{A}_{i}$ and $\mathcal{N}^{B}_{i}$ becomes zero.
This corresponds to the case where $a=0$ or $b=0$.
In fact, if $c=2^{n}$ for $n\in\mathbb{Z}$, the rank of $\tilde{C}$ becomes exactly 1.
Furthermore, since the singular values $\sigma_{n}$ of the matrix $\tilde{C}$ can be approximated as
\begin{align}
\sigma_{1}\simeq&
1+\frac{1}{2}p^{2(c-a)}(1-p^{2a})^2,\\
\sigma_{2}\simeq&
p^{c}\sqrt{1-p^{2a}}
\end{align}
for $N_{A}\gg 1$, the singular values exponentially decay to zero for large cutoff $c$ and the rank of $\tilde{C}$ can be approximated to 1.

\subsubsection{Singular Values Corresponding to a Hann Window Function}
\label{app:orthogonal_center_for_hann_function}

Finally, we discuss the case where the singular values can be written as a linear combination of exponential functions.
In this section, we introduce the Hann function $h(n)$, which is known as a window function in the context of signal processing, and consider its quantum circuit representation.
The Hann function can be expressed as
\begin{align}
\sqrt{\mathcal{N}}w_{n}=&h(n)
=\frac{1}{2}\left[1+\cos{\left(\frac{\pi n}{2^{N}}\right)}\right]
\nonumber\\=&
\frac{1}{2}+\frac{1}{4}\sum_{\sigma=\pm 1}e^{i\sigma\alpha n},
\end{align}
where $\mathcal{N}$ is the normalization factor and $\alpha=\frac{\pi}{2^{N}}$.
The profile of this distribution is illustrated in the right panel of Fig.~\ref{fig:qc-svd-examples}(c).
In the following, we assume that the integers $N_{A}$, $N_{B}$, and $N$ are sufficiently large.
We note that the normalization factor $\mathcal{N}$ evaluates to $\mathcal{N}=2^{N}\beta$ with $\beta\simeq\frac{3}{8}$ using the Riemann sum formula.
Decomposing the global index as $l=n+2^{N_{A}}m$ separates the phase factor into contributions from lower bits $n$ ($0\le n< 2^{N_{A}}$) and upper bits $m$ ($0\le m< 2^{N_{B}}$).
Consequently, the state $\hat{\Omega}|0\rangle$ with coefficients corresponding to the Hann function can be represented as:
\begin{align}
\sqrt{\mathcal{N}}\hat{\Omega}|0\rangle=&
\frac{1}{2}\sum^{2^{N_{A}}-1}_{n=0}|n\rangle
\otimes\sum^{2^{N_{B}}-1}_{m=0}|m\rangle
\nonumber\\&
+\frac{1}{4}\sum_{\sigma=\pm 1}\left(\sum^{2^{N_{A}}-1}_{n=0}e^{i\sigma\alpha n}|n\rangle\right)
\nonumber\\&
\otimes\left(\sum^{2^{N_{B}}-1}_{m=0}e^{i\sigma\alpha 2^{N_{A}}m}|m\rangle\right).
\end{align}
Next, we choose the normalized bases $|a_{\mu}\rangle$ and $|b_{\mu}\rangle$~($\mu=-1,0,1$) for the subsystems $A$ and $B$ as
\begin{align}
|a_{\mu}\rangle=&\sqrt{\frac{1}{2^{N_{A}}}}
\sum^{2^{N_{A}}-1}_{n=0}e^{i\mu\alpha n}|n\rangle,\\
|b_{\mu}\rangle=&\sqrt{\frac{1}{2^{N_{B}}}}
\sum^{2^{N_{B}}-1}_{n=0}e^{i\mu\alpha 2^{N_{A}}n}|n\rangle.
\end{align}
To orthogonalize the above bases, we define the following two overlap matrices:
\begin{align}
[S_{A}]_{\mu,\mu'}=&\langle a_{\mu}|a_{\mu'}\rangle
=\frac{1}{2^{N_{A}}}\sum^{2^{N_{A}}-1}_{n=0}
e^{-i(\mu-\mu')\alpha n},\\
[S_{B}]_{\mu,\mu'}=&\langle b_{\mu}|b_{\mu'}\rangle
=\frac{1}{2^{N_{B}}}\sum^{2^{N_{B}}-1}_{n=0}
e^{-i(\mu-\mu')\alpha 2^{N_{A}}n}
\end{align}
where the matrices $S_{A}$ and $S_{B}$ are positive semi-definite and Hermitian.
In particular, when $N_{A}$ and $N_{B}$ are sufficiently large, we approximate the matrices $S_{A}$ and $S_{B}$ by expanding the exponential terms and retaining the relevant leading-order terms:
\begin{align}
S_{A}\simeq&\begin{pmatrix}1&1&1\\1&1&1\\1&1&1\end{pmatrix}
+i\frac{\pi}{2}2^{-N_{B}}
\begin{pmatrix}0&1&2\\-1&0&1\\-2&-1&0\end{pmatrix}
\nonumber\\&
-\frac{\pi^2}{6}2^{-2N_{B}}
\begin{pmatrix}0&1&4\\1&0&1\\4&1&0\end{pmatrix}
,\\
S_{B}\simeq&I_{3}+i\frac{2}{\pi}
\begin{pmatrix}0&1&0\\-1&0&1\\0&-1&0\end{pmatrix}
\end{align}
This result suggests that the matrix $S_{A}$ becomes rank-deficient for $N_{A},N_{B}\gg 1$.
Indeed, perturbative analysis reveals that the leading non-vanishing contribution to the smallest eigenvalue of $S_{A}$ is of the order of $\mathcal{O}(2^{-4N_{B}})$, decreasing exponentially with increasing $N_{B}$.
This implies that the state on the smaller subsystem $A$ can be mainly composed of at most two orthonormal basis vectors, and the neglected components do not affect the subsystem $A$ corresponding to the lower qubits.
Next, using the unitary matrices $U_{A}$ and $U_{B}$, and the diagonal matrices $D_{A}$ and $D_{B}$, we diagonalize the matrices $S_{A}$ and $S_{B}$ as follows:
\begin{align}
&S_{A} = U_{A}D_{A}U_{A}^{\dagger},\\
&S_{B} = U_{B}D_{B}U_{B}^{\dagger}.
\end{align}
Then, we can choose the orthonormalized bases for the subsystems $A$ and $B$ as 
\begin{align}
|\varphi^{A}_{n}\rangle=&\sqrt{\frac{1}{d^{A}_{n}}}
\sum_{\mu}|a_{\mu}\rangle[U_{A}]_{\mu,n},\\
|\varphi^{B}_{n}\rangle=&\sqrt{\frac{1}{d^{B}_{n}}}
\sum_{\mu}|b_{\mu}\rangle[U_{B}]_{\mu,n},
\end{align}
where we denote the $n$-th entry of the diagonal matrices $D_{A}$ and $D_{B}$ as $d^{A}_{n}$ and $d^{B}_{n}$, respectively.
Thus, the state $\hat{\Omega}|0\rangle$ can be represented as
\begin{align}
\hat{\Omega}|0\rangle=&
\sqrt{\frac{1}{\beta}}\sum_{\mu}C_{\mu}|a_{\mu}\rangle\otimes|b_{\mu}\rangle
\nonumber\\=&
\sum_{n,m}\tilde{C}_{n,m}|\varphi^{A}_{n}\rangle
\otimes|\varphi^{B}_{m}\rangle
\end{align}
where the coefficients are given by $C_{0}=1/2$ and $C_{\pm 1}=1/4$.
Using the coefficients $C_{\mu=0,\pm 1}$, the coefficients $\tilde{C}_{n,m}$ can be written as follows.
\begin{align}
\tilde{C}_{n,m}=\sqrt{\frac{d^{A}_{n}d^{B}_{m}}{\beta}}
\sum_{\mu}C_{\mu}[U_{A}]^{*}_{\mu,n}[U_{B}]^{*}_{\mu,m}
\end{align}
Performing the SVD on the matrix $\tilde{C}_{n,m}$ reveals that the rank of $\tilde{C}$ is at most 3, and it reduces to 2 for small $N_{A}$ as shown in Fig.~\ref{fig:qc-svd-examples}(d).
This reduction arises because, for small $N_{A}$, the phase variation $e^{i\alpha n}$ is insufficient to fully distinguish between the three basis states $|a_{\mu}\rangle$, resulting in their effective linear dependence.
This implies that the state $\hat{\Omega}|0\rangle$ can be implemented as a quantum circuit composed of sequentially aligned unitary gates with a gate size of $3$, as depicted in the left panel of Fig.~\ref{fig:qc-svd-examples}(c).
Indeed, numerical results in Figs.~\ref{fig:qc-svd-examples}(e) and (f) demonstrate that as the gate size $D$ approaches 3, the infidelity vanishes and the generated distribution $w_{n}$ converges to the exact one.
While this discussion focuses on the Hann function, the underlying principle extends to any function composed of a finite number of Fourier components.
As demonstrated, the required bond dimension of the MPS--corresponding to the gate size--is bounded by the number of Fourier components included in the target function.
Moreover, the proposed method for constructing circuits with arbitrary weights has applications beyond reproducing singular value distributions.
Crucially, this construction scheme is applicable to quantum algorithms for generating random numbers following arbitrary probability distributions.

\subsection{Application to SVD for General Matrix Operators}
\label{app:svd_for_general_matrix}

In our study, we have focused on the variational quantum algorithm for the Schmidt decomposition of a bipartite quantum state.
Since the Schmidt decomposition is mathematically equivalent to SVD, it is instructive to discuss the application of our quantum algorithm to SVD for general matrix operators.
In this section, to demonstrate the applicability of our method, we show that the SVD of a general matrix operator $\hat{A}$ can be mapped to the Schmidt decomposition of a corresponding bipartite quantum state via a vectorization process.

First, for simplicity, we assume that the general matrix operator $\hat{A}$ can be decomposed into a linear combination of unitary gates $\hat{G}_{i}$ as $\hat{A}=\sum_{i}\gamma_{i}\hat{G}_{i}$ with complex coefficients $\gamma_{i}$.
To transform this operator into a bipartite quantum state, we consider a system where each subsystem consists of $K$ qubits, assuming a square matrix dimension of $2^{K}$.
We introduce a quantum circuit $\hat{W}$ constructed from a product of $R_{y}$ gates (or Hadamard gates): $\bigotimes^{K}_{i=1}R_{y}(\pi/2)$, and a cascade of CNOT gates, which prepares a maximally entangled state $\hat{W}|0,0\rangle=2^{-K/2}\sum^{2^{K}-1}_{n=0}|n,n\rangle$.
By applying the unitary gates $\hat{G}_{i}$ to the first subsystem of this state, we obtain the following unnormalized relation:
\begin{align}
|A\rangle\equiv&
\sum_{i}\gamma_{i}
(\hat{G}_{i}\otimes\hat{I})\hat{W}|0,0\rangle
\nonumber\\=&
2^{-K/2}\sum_{n}\sum_{i}\gamma_{i}(\hat{G}_{i}\otimes\hat{I})|n,n\rangle
\nonumber\\=&
2^{-K/2}\sum_{m,n}A_{m,n}|m,n\rangle,
\end{align}
where we have used the definition $\hat{A}=\sum_{m,n}A_{m,n}|m\rangle\langle n|$.

This result demonstrates that the vectorized form $|A\rangle$ of the operator $\hat{A}$ can be explicitly constructed using the quantum circuit $\hat{W}$ and a linear combination of the unitary gates $\hat{G}_{i}$ [Fig.~\ref{fig:qc-operator-vectorize}(a)].
Therefore, the SVD of a general matrix operator can, at the formal level, be mapped to the Schmidt decomposition problem with a bipartite ansatz $|B\rangle=\sum_{n}w_{n}\hat{U}\otimes\hat{V}|n,n\rangle$ [Fig.~\ref{fig:qc-operator-vectorize}(b)].
In practice, however, the overlaps required for our variational SVD algorithms rely on the ability to realize the target state $|A\rangle$ as a single state-preparation procedure with post-selection.
Crucially, unlike standard methods that require the sequential quantum evaluation of quantities such as $\langle A|B\rangle=\sum_{n}w_{n}[U^{\dagger}AV^{*}]_{n,n}$ as used in \cite{wang_Variational_2021}, this formulation allows us to apply unitary operations in parallel to the two subsystems during the state preparation.
Nevertheless, this structural advantage is realized only when the vectorized state $|A\rangle$ admits an efficient single-circuit implementation, for example through a block-encoding scheme with post-selection or another compact circuit embedding of the operator sum.
Otherwise, the evaluation of the relevant overlaps reduces to term-by-term processing over the operator decomposition of $\hat{A}$, and the practical advantage of the present framework is substantially diminished.
Thus, the vectorization viewpoint provides a natural conceptual route from operator SVD to Schmidt decomposition, while its practical utility depends on whether such an efficient target-state preparation is available.


\begin{figure}[htbp]
\centering
\includegraphics[width=0.9\columnwidth]{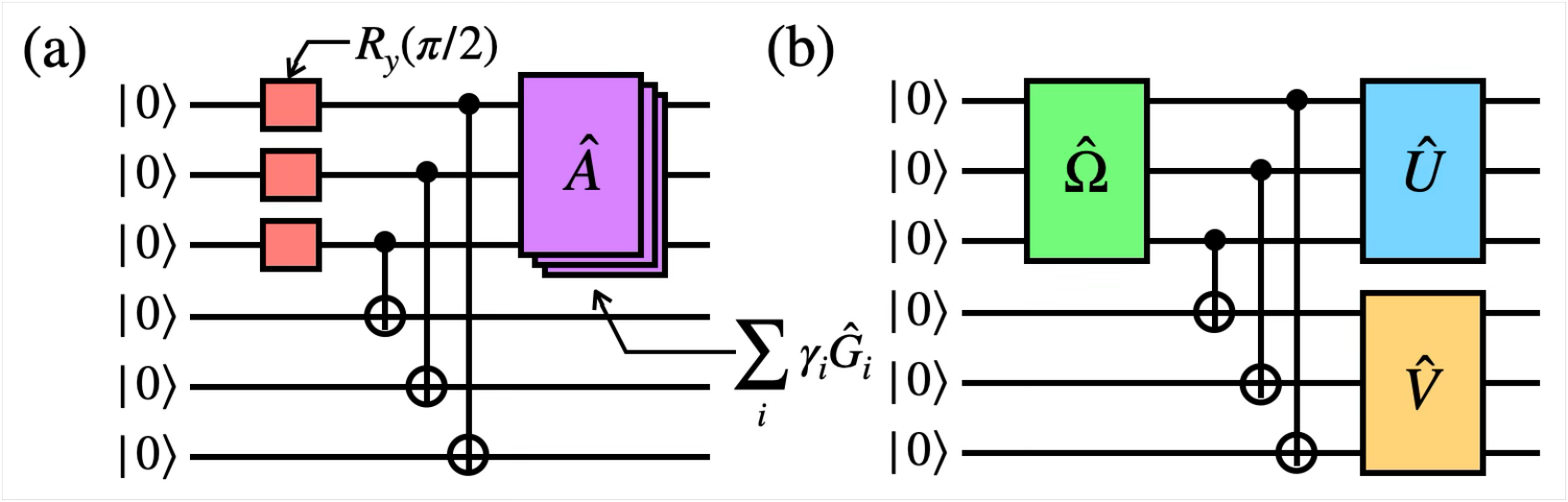}
\caption{
(a) Quantum circuit for preparing a state proportional to the vectorized state $|A\rangle=\sum_{n,m}|n,m\rangle\langle m|\hat{A}|n\rangle$, corresponding to a general matrix operator $\hat{A}=\sum_{i}\gamma_{i}\hat{G}_{i}$.
The red boxes indicate the product of $R_{y}(\pi/2)$ gates.
(b) Quantum circuit for the bipartite ansatz $|B\rangle=\sum_{n}w_{n}\hat{U}\otimes\hat{V}|n,n\rangle$ used in our variational SVD algorithms.
}
\label{fig:qc-operator-vectorize}
\end{figure}


\section{Mathematical Details of Variational SVD Algorithms}
\label{app:mathematical_details_of_vqsvd}

\subsection{Majorization}
\label{app:majorization}

In this section, we introduce the concept of majorization and explain the Hardy-Littlewood-P\'olya theorem, and discuss its application in the subsequent section.
The details of majorization, including the proofs of the theorems, can be found in the monograph by Marshall, Olkin, and Arnold \cite{marshall_Inequalities_2011a}.

\begin{dfn}[Majorization]
\label{thm:majorization}
Given a vector $\bm{x}=(x_{1},\cdots,x_{N})\in\mathbb{R}^{N}$, we denote by $\bm{x}^{\downarrow}$ the vector with the same entries as $\bm{x}$ rearranged in decreasing order: $x^{\downarrow}_{1}\ge\cdots\ge x^{\downarrow}_{N}$.
The vector $\bm{x}$ is said to be majorized by $\bm{y}$ (denoted by $\bm{x}\prec\bm{y}$) if
\begin{align}
\sum^{k}_{n=1}x_{n}^{\downarrow}\le \sum^{k}_{n=1}y_{n}^{\downarrow}
\end{align}
for $k=1,\cdots,N-1$ and
\begin{align}
\sum^{N}_{n=1}x_{n}^{\downarrow}=\sum^{N}_{n=1}y_{n}^{\downarrow}.
\end{align}
\end{dfn}

\begin{dfn}[Doubly Stochastic Matrix]
\label{def:doubly_stochastic}
A matrix $W=(W_{n,m})_{1\le n,m\le N}\in\mathbb{R}^{N\times N}$ is called doubly stochastic if all entries are non-negative ($W_{n,m}\ge 0$) and each row and column sums to 1:
\begin{align}
\sum^{N}_{m=1}W_{n,m}=1\quad(\forall n),\\
\sum^{N}_{n=1}W_{n,m}=1\quad(\forall m).
\end{align}
\end{dfn}

\begin{thm}[Hardy-Littlewood-P\'olya theorem]
\label{thm:hardy-littlewood-polya}
Let $\bm{x}$ and $\bm{y}$ be two vectors in $\mathbb{R}^{N}$.
Then the following statements are equivalent:
\begin{enumerate}[label=(\Alph*)]
\item $\bm{x}\prec\bm{y}$,
\item there exists a doubly stochastic matrix $W=(W_{n,m})_{1\le n,m\le N}$ such that $\bm{x}=W\bm{y}$,
\item the inequality $\sum^{N}_{n=1}f(x_{n})\le \sum^{N}_{n=1}f(y_{n})$ holds for every continuous convex function $f(x)$.
\end{enumerate}
\end{thm}
The formulation of the above definition and theorem concerning the majorization in our paper follows Ref.~\cite{marshall_Inequalities_2011a}.
Furthermore, utilizing the above theorem, we can prove the following properties for the weighted majorization.

\begin{cor}
\label{thm:weighted-majorization}
Let $\bm{x}$ and $\bm{y}$ be two vectors in $\mathbb{R}^{N}$ satisfying $\bm{x}\prec\bm{y}$ with decreasing entries.
Let $\bm{c}$ be the vector in $\mathbb{R}^{N}$ with entries arranged in decreasing order.
Suppose $f(x)$ is continuous, convex, and monotonically increasing.
Then the following inequality holds:
\begin{align}
\sum^{N}_{n=1}c_{n}f(x_{n})\le \sum^{N}_{n=1}c_{n}f(y_{n}).
\end{align}
Similarly, if $f(x)$ is continuous, convex, and monotonically decreasing, the following inequality holds:
\begin{align}
\sum^{N}_{n=1}c_{n}f(x_{n})\le \sum^{N}_{n=1}c_{N-n+1}f(y_{n}).
\end{align}
\end{cor}

\begin{proof}
From Theorem~\ref{thm:hardy-littlewood-polya}, there exists a doubly stochastic matrix $W=(W_{nm})_{1\le n,m\le N}$ such that $\bm{x}=W\bm{y}$.
Using Jensen's inequality for the convex function $f$, we have:
\begin{align}
\sum^{N}_{n=1}c_{n}f(x_{n})=&
\sum^{N}_{n=1}c_{n}f(\sum^{N}_{m=1}W_{n,m}y_{m})
\nonumber\\\le&
\sum^{N}_{n=1}c_{n}\sum^{N}_{m=1}W_{n,m}f(y_{m})
\nonumber\\=&
\sum^{N}_{m=1}\left(\sum^{N}_{n=1}c_{n}W_{n,m}\right)f(y_{m})
\nonumber\\=&
\sum^{N}_{m=1}d_{m}f(y_{m})
\end{align}
where we define $d_{m}=\sum^{N}_{n=1}W_{n,m}c_{n}$.
Since $W$ is doubly stochastic, this relation implies $\bm{d}\prec\bm{c}$.
Since $f$ is monotonically increasing and $\bm{y}$ is decreasing, the sequence $\{f(y_{n})\}$ is decreasing.
Given $\bm{d}\prec\bm{c}$ and the ordering of $\bm{c}$ and $\{f(y_{n})\}$, we can apply the property of majorization to obtain:
\begin{align}
\sum^{N}_{n=1}d_{n}f(y_{n})\le&
\sum^{N}_{n=1}d_{n}^{\downarrow}f(y_{n})
\nonumber\\=&
\sum^{N}_{n=1}d_{n}^{\downarrow}f(y_{N})
\nonumber\\&
+\sum^{N-1}_{n=1}(f(y_{n})-f(y_{n+1}))
\sum^{n}_{m=1}d_{m}^{\downarrow}
\nonumber\\\le&
\sum^{N}_{n=1}c_{n}^{\downarrow}f(y_{n})
=\sum^{N}_{n=1}c_{n}f(y_{n})
\end{align}
Similarly, if $f(x)$ is monotonically decreasing, $\{f(y_{n})\}$ becomes increasing.
The upper bound is then achieved by pairing the largest $c$ with the largest $f(y)$, which corresponds to the reverse order:
\begin{align}
\sum^{N}_{n=1}d_{n}f(y_{n})\le&
\sum^{N}_{n=1}c_{N-n+1}f(y_{n})
\end{align}
\end{proof}
As seen from Theorem~\ref{thm:weighted-majorization}, if we assume $c_{n}=1$ for every $n$, we recover the property of majorization in Theorem~\ref{thm:hardy-littlewood-polya}(C).
In the subsequent section, we prove the fundamental inequality used in our SVD algorithms by utilizing the above theorems.

\subsection{Upper Bound of the Objective Function for Full Optimization}
\label{app:details_of_global_optimization}

In this section, utilizing the concept of majorization discussed in the previous section, we discuss the properties of the full optimization algorithm by evaluating the upper bound of the objective function, which corresponds to the overlap between the reference state and the quantum circuit state.

First, we assume that the reference state $\ket{\Psi}$ and the quantum circuit state $\ket{\Phi}$ are explicitly written as follows:
\begin{align}
\ket{\Psi}=&\sum^{N}_{n=1}\sum^{M}_{m=1}C_{n,m}\ket{n,m},
\\
\ket{\Phi}=&(\hat{U}\otimes\hat{V})\hat{W}|0\rangle=\sum^{K}_{n=1}w_{n}\ket{u_{n},v_{n}},
\end{align}
with $\hat{W}|0\rangle=\sum^{K}_{n=1}w_{n}|n,n\rangle$, $|u_{n}\rangle=\hat{U}|n\rangle$, $|v_{n}\rangle=\hat{V}|n\rangle$, and $K=\mathrm{min}(N,M)$.
$N$ and $M$ represent the dimensions of the left and right subspaces, respectively.
We denote the product state $\ket{\phi}\otimes\ket{\varphi}$ as $\ket{\phi,\varphi}$.
The coefficients $\{w_{n}\}^{K}_{n=1}$ are non-negative numbers defined in decreasing order: $1\ge w_{1}\ge\cdots\ge w_{K}\ge 0$, and satisfy the normalization condition $\sum^{K}_{n=1}w^{2}_{n}=1$.
Under the above notations, the overlap between $\ket{\Psi}$ and $\ket{\Phi}$ is given by
\begin{align}
\langle\Phi|\Psi\rangle=&\sum^{K}_{k=1}\sum^{N}_{n=1}\sum^{M}_{m=1}
w_{k}\langle u_{k}|n\rangle C_{n,m}\langle v_{k}|m\rangle
\nonumber\\=&
\mathrm{Tr}\left[WU^{\dagger}CV^{*}\right],
\end{align}
where $W$ is a $M\times N$ rectangular matrix with diagonal entries $W_{n,n}=w_{n}$ and zero off-diagonal entries, and $U$ and $V$ are $N\times N$ and $M\times M$ unitary matrices defined by $U_{k,n}=\langle k|u_{n}\rangle$ and $V_{m,k}=\langle m| v_{k}\rangle$, respectively.
\begin{widetext}
Moreover, we can rewrite the real part of $\langle\Phi|\Psi\rangle$ using a block matrix representation:
\begin{align}
\mathrm{Re}\langle\Phi|\Psi\rangle=&
\frac{1}{2}\mathrm{Tr}\left[
\begin{pmatrix}0&W\\W&0\end{pmatrix}
\begin{pmatrix}U^{\dagger}&0\\0&V^{T}\end{pmatrix}
\begin{pmatrix}0&C\\C^{\dagger}&0\end{pmatrix}
\begin{pmatrix}U&0\\0&V^{*}\end{pmatrix}
\right]
\nonumber\\=&
\frac{1}{2}\mathrm{Tr}\left[
\begin{pmatrix}0&W\\W&0\end{pmatrix}
\begin{pmatrix}U^{\dagger}X&0\\0&V^{T}Y\end{pmatrix}
\begin{pmatrix}0&D\\D&0\end{pmatrix}
\begin{pmatrix}X^{\dagger}U&0\\0&Y^{\dagger}V^{*}\end{pmatrix}
\right]
\nonumber\\=&
\frac{1}{2}\mathrm{Tr}\left[\overline{W}Q^{\dagger}\overline{D}Q\right]
\end{align}
where the SVD of $C$ is given by $C=XDY^{\dagger}$, and the entries of the diagonal matrix $D$ satisfy $1\ge \sigma_{1}\cdots\ge \sigma_{K}\ge 0$.
\end{widetext}
We assume the matrices $\overline{W}$ and $\overline{D}$ are diagonal matrices with entries $\bar{w}_{n}\in\bigcup^{K}_{n=1}\{w_{n},-w_{n}\}$ and $\bar{\sigma}_{n}\in\bigcup^{K}_{n=1}\{\sigma_{n},-\sigma_{n}\}$, arranged in decreasing order.
Specifically:
\begin{align}
\overline{W}=&\begin{pmatrix}W&0\\0&-PWP\end{pmatrix},
\\
\overline{D}=&\begin{pmatrix}D&0\\0&-PDP\end{pmatrix},
\end{align}
where $P$ is a permutation matrix satisfying $P_{n,m}=\delta_{n,K-m+1}$.
Additionally, we define the unitary matrix $Q$ as
\begin{align}
Q=\left[
\begin{matrix}
\frac{1}{2}(X^{\dagger}U+Y^{\dagger}V^{*})&\frac{i}{2}P(X^{\dagger}U-Y^{\dagger}V^{*})\\
-\frac{i}{2}(X^{\dagger}U-Y^{\dagger}V^{*})P&\frac{1}{2}P(X^{\dagger}U+Y^{\dagger}V^{*})P
\end{matrix}
\right].
\end{align}
Next, let $\bar{\sigma}'_{n}$ be the diagonal elements of the matrix $Q^{\dagger}\overline{D}Q$.
Since the matrix $R_{n,m}=|Q_{n,m}|^2$ is a doubly stochastic matrix satisfying $\sum^{2K}_{n=1}R_{n,m}=\sum^{2K}_{m=1}R_{n,m}=1$ and $\bar{\sigma}'_{n}=\sum^{2K}_{m=1}R_{n,m}\bar{\sigma}_{m}$, we find that $\bar{\bm{\sigma}}$ majorizes $\bar{\bm{\sigma}}'$ by Theorem \ref{thm:hardy-littlewood-polya}.
Thus, we obtain the following inequality:
\begin{align}
\mathrm{Re}\langle\Phi|\Psi\rangle
&= \frac{1}{2}\sum^{2K}_{n=1}\bar{w}_{n}\bar{\sigma}'_{n}
\nonumber\\&
\le \frac{1}{2}\sum^{2K}_{n=1}\bar{w}_{n}\bar{\sigma}_{n}
= \sum^{K}_{n=1}w_{n}\sigma_{n}.
\label{eq:inequality-fidelity}
\end{align}
where we apply Corollary~\ref{thm:weighted-majorization} with $f(x)=x$, demonstrating that the weighted sum is maximized when the vectors are similarly ordered.
It is worth noting that the equality condition is given by
\begin{align}
X^{\dagger}U=Y^{\dagger}V^{*}=\Lambda
\end{align}
where $\Lambda$ is the diagonal matrix with the entries $e^{i\phi_{n}}$.
Thus, noting that the diagonal matrices $\Lambda$ and $D$ commute, we find that the maximization of $\mathrm{Re}\langle\Phi|\Psi\rangle$ for arbitrary matrices $U$ and $V$ corresponds to the SVD: $C=UDV^{T}=X\Lambda D\Lambda^{-1}Y^{\dagger}=XDY^{\dagger}$.

On the other hand, if the unitary matrices $U$ and $V$, represented by a limited number of variational parameters, do not have enough expressive power to perform the SVD exactly, we cannot obtain the rigorous singular values by simply solving the maximization problem of $\mathrm{Re}\langle\Phi|\Psi\rangle$.
In such cases, we need to consider alternative approaches for evaluating the approximated singular values using the observables $c_{n}=(U^{\dagger}_\mathrm{opt}CV^{*}_\mathrm{opt})_{n,n}$ with the optimized unitary matrices $U=U_\mathrm{opt}$ and $V=V_\mathrm{opt}$.
We focus on the following inequality:
\begin{align}
\sum^{K}_{n=1}w_{n}\mathrm{Re}c_{n}&
\le\sum^{K}_{n=1}w_{n}|\mathrm{Re}c_{n}|
\nonumber\\&
\le\sum^{K}_{n=1}w_{n}|c_{n}|
\le\sum^{K}_{n=1}w_{n}|c_{\tau_{n}}|
\label{eq:inequality-real-to-abs}
\end{align}
where $\tau_{n}$ is a permutation of $\{1,\cdots,K\}$ such that $|c_{\tau_{n}}|$ are sorted in decreasing order, and we assume that $w_{n}$ are non-negative and also sorted in decreasing order.
Next, when we introduce the unitary matrices $U$ and $V$ such that $U_{n,m} = \sum^{K}_{k=1}[U_\mathrm{opt}]_{n,k}e^{i\mathrm{arg}(c_{k})}\delta_{k,\tau_{m}}$ and $V_{n,m} = \sum_{k}[V_\mathrm{opt}]_{n,k}\delta_{k,\tau_{m}}$, we obtain the following inequality:
\begin{align}
\mathrm{Re}\mathrm{Tr}(WU^{\dagger}CV^{*})=&
\sum^{K}_{n=1}w_{n}e^{-i\mathrm{arg}(c_{\tau_{n}})}c_{\tau_{n}}
\nonumber\\=&
\sum^{K}_{n=1}w_{n}|c_{\tau_{n}}|
\le\sum^{K}_{n=1}w_{n}\sigma_{n}
\label{eq:inequality-abs}
\end{align}
where we used Eq.~\eqref{eq:inequality-fidelity}.
Thus, combining Eqs.~\eqref{eq:inequality-real-to-abs}, and \eqref{eq:inequality-abs}, we obtain the following inequality:
\begin{align}
\mathrm{Re}\langle\Phi|\Psi\rangle
\le\sum^{K}_{n=1}w_{n}|c_{\tau_{n}}|
\le\sum^{K}_{n=1}w_{n}\sigma_{n}.
\end{align}
This result suggests that the absolute values $\{|c_{n}|\}^{K}_{n=1}$ sorted in the descending order provide the tightest upper bound of $\mathrm{Re}\langle\Phi|\Psi\rangle$ and serve as the best approximation for the exact singular values $\{\sigma_{n}\}^{K}_{n=1}$.

\subsection{Error Propagation Analysis for the Partial Optimization Algorithm}
\label{app:details_of_partial_optimization}

In our partial optimization algorithm, we iteratively recover the singular values $\{\sigma_{n}\}$ from the cumulative sum of their $p$-th powers.
We specifically choose $p=1$, as this choice is critical for ensuring numerical robustness and minimizing error amplification.
In this subsection, we analyze the error propagation for various values of $p$ to justify this choice.
First, let $s_{k}^{(n)}=\sigma_{k}+\delta s_{k}^{(n)}$ denote the estimated $k$-th singular value obtained at the $n$-th optimization step, where $\delta s_{k}^{(n)}$ represents the estimation error.
Using these values, we define the cumulative sum of the $p$-th powers, denoted as $S_{n}$, as follows:
\begin{align}
&S_{n} = \sum^{n}_{k=1}(s_{k}^{(n)})^{p}.
\end{align}
Subsequently, the $n$-th singular value is recovered from the difference between consecutive sums:
\begin{align}
s_{n}^{p}=S_{n}-S_{n-1}.
\end{align}
Based on these definitions, assuming that the errors are sufficiently small (i.e., $|\delta s^{(n)}_{k}|\ll\sigma_{k}$), the absolute error of the cumulative sum, $\delta S_{n}$, and the error in the recovered singular value, $\delta s_{n}$, can be expressed to the first order as:
\begin{align}
\delta S_{n}=&
\sum^{n}_{k=1}p(\sigma_{k})^{p-1}\delta s_{k}^{(n)},
\label{eq:absolute-error-cumulative-sum}\\
\delta s_{n}=& 
\frac{1}{p}\sigma_{n}^{1-p}(\delta S_{n}-\delta S_{n-1}).
\label{eq:absolute-error-recovered-singular-value}
\end{align}
Combining these two relations, the relative error $\delta s_{n}/\sigma_{n}$ is given by
\begin{align}
\frac{\delta s_{n}}{\sigma_{n}}=&
\frac{\delta s_{n}^{(n)}}{\sigma_{n}}
+\sum^{n-1}_{k=1}\left(\frac{\sigma_{k}}{\sigma_{n}}\right)^{p}
\left(\frac{\delta s_{k}^{(n)}}{\sigma_{k}}-\frac{\delta s_{k}^{(n-1)}}{\sigma_{k}}\right).
\label{eq:relative-error-recovered-singular-value}
\end{align}
We now analyze the numerical stability of this scheme, considering the order of singular values $\sigma_{k}\ge\sigma_{n}$ for $k<n$.
For the case of $0\le p< 1$, the exponent $p-1$ in Eq.~\eqref{eq:absolute-error-cumulative-sum} is negative.
Consequently, the term corresponding to $k=n$, $p(\sigma_{n})^{p-1}\delta s_{n}^{(n)}$, diverges as the smallest singular value $\sigma_{n}$ approaches zero.
This singularity causes the absolute error $\delta S_{n}$ to amplify significantly, rendering the forward calculation of the cumulative sum unstable.
Conversely, for $p>1$, we examine the relative error propagation in Eq.~\eqref{eq:relative-error-recovered-singular-value}.
Here, the amplification factor for the errors from previous steps ($k<n$) is given by the ratio $(\sigma_{k}/\sigma_{n})^{p}$.
Since $\sigma_{k}\ge\sigma_{n}$, this amplification factor grows superlinearly with respect to this ratio for $p>1$, leading to a significantly stronger accumulation of errors compared to the linear case of $p=1$.
Thus, $p=1$ is the optimal choice as it minimizes this amplification while avoiding the divergence issue associated with $p<1$.
For these reasons, we adopt $p=1$ to ensure numerical robustness.


\section{Benchmark on a 1D Gapped System}
\label{app:benchmark_on_1d_gapped_system}

\subsection{1D Heisenberg Spin Ladder}
\label{app:benchmark_on_1d_spin_ladder}


\begin{figure*}[htbp]
\centering
\includegraphics[width=2.0\columnwidth]{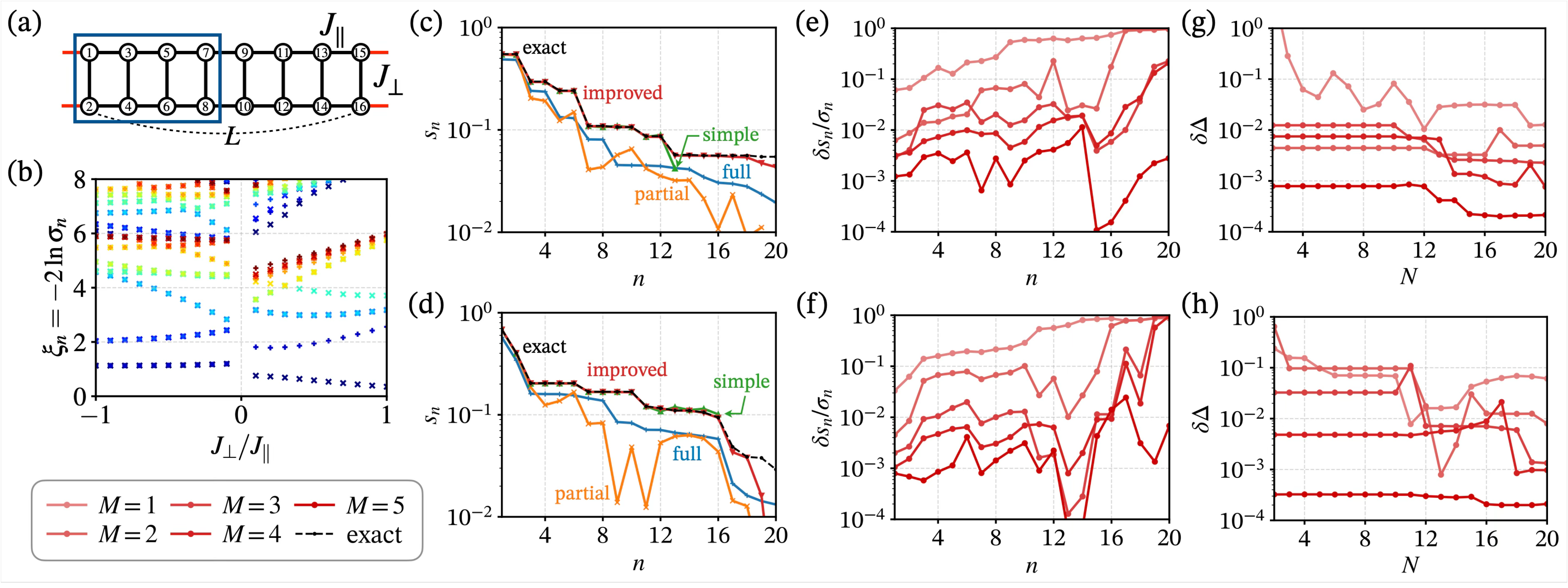}
\caption{
Performance benchmarks of partial SVD algorithms for the 1D Heisenberg spin ladder with system size $L=8$ under antiperiodic boundary conditions (APBC).
The number of layers is fixed to $M=4$ in (c) and (d).
(a) Schematic representation of the 1D spin ladder, where each site is labeled by its site index and the rectangular region enclosed by a blue line denotes the subsystem $A$.
(b) Entanglement spectrum $\xi_{n}=-2\ln{\sigma_{n}}$ as a function of the ratio $J_{\perp}/J_{\parallel}$.
The double degeneracy observed for $J_{\perp}/J_{\parallel}<0$ signifies a topologically non-trivial ground state protected by inversion symmetry.
(c)-(h) Comparison of algorithm performance for the topologically non-trivial ($J_{\perp}/J_{\parallel}=-0.1$; top panels) and trivial ($J_{\perp}/J_{\parallel}=0.1$; bottom panels) phases.
(c)(d) Singular values obtained via the four proposed partial SVD algorithms.
(e)-(h) Results obtained via the improved deflation algorithm using the quantum circuit with the number of layers up to $M=5$.
(e)(f) Relative errors in singular values.
(g)(h) Absolute error in the Schmidt gap $\Delta=\xi_{2}-\xi_{1}$ as a function of deflation steps $N$.
}
\label{fig:vqsvd-demo-ladder}
\end{figure*}


While the main text focuses on representative gapless systems, such as the 1D chain and 2D square lattice, this section benchmarks our variational SVD algorithms on a prototypical gapped system: the 1D Heisenberg spin ladder shown in Fig.~\ref{fig:vqsvd-demo-ladder}(a).
The Hamiltonian of the two-leg spin ladder is given by
\begin{align}
\hat{H}=& J_{\parallel}\sum_{l=1,2}\sum^{L}_{n=1}
\hat{\bm{S}}_{l,n}\cdot\hat{\bm{S}}_{l,n+1}
+J_{\perp}\sum^{L}_{n=1}\hat{\bm{S}}_{1,n}\cdot\hat{\bm{S}}_{2,n},
\end{align}
where $\hat{\bm{S}}_{l,n}$ is the spin-$1/2$ operator acting on a qubit located at a 2D coordinate $(l,n)$, with $l\in\{1,2\}$ denoting the leg index and $n\in\{1,\cdots,L\}$ the rung index.
The explicit mapping between these physical 2D coordinates and the qubit site indices is illustrated in Fig.~\ref{fig:vqsvd-demo-ladder}(a).
The parameters $J_{\parallel}$ and $J_{\perp}$ represent the exchange couplings along the legs and rungs, respectively.
To investigate the topological properties of the ground state, we impose twisted boundary conditions along the leg direction, which is implemented by setting $\hat{S}^{\pm}_{l,L+1}=e^{\pm i\theta}\hat{S}^{\pm}_{l,1}$ and $\hat{S}^{z}_{l,L+1}=\hat{S}^{z}_{l,1}$.

In our calculations, we fix the system length to $L=8$ and the leg coupling to $J_{\parallel}=1$.
By tuning the ratio of the rung to leg couplings, $J_{\perp}/J_{\parallel}$, we investigate two distinct gapped ground states.
In the positive regime $J_{\perp}/J_{\parallel}=0.1$, the ground state forms a trivial rung-singlet phase.
Conversely, the negative regime $J_{\perp}/J_{\parallel}=-0.1$ yields a Haldane phase composed of $S=1$ triplets.
Crucially, under antiperiodic boundary conditions (APBC) implemented by $\theta=\pi$, the Haldane phase is a topologically non-trivial state characterized by a strict double degeneracy in its entanglement spectrum (ES), $\xi_{n}=-2\ln{\sigma_{n}}$, which is protected by inversion symmetry \cite{pollmann_Entanglement_2010,miyakoshi_Entanglement_2016}.
As demonstrated in Fig.~\ref{fig:vqsvd-demo-ladder}(b), this topological degeneracy provides a stringent physical benchmark for our partial SVD algorithms.

Figures~\ref{fig:vqsvd-demo-ladder}(c) and (d) display the singular values computed by the four proposed algorithms for these two ground states.
The results reveal a contrast in performance: while global optimization approaches (full and partial optimization) fail to resolve the exact two-fold degeneracy in the Haldane phase, the simple and improved deflation algorithms precisely reproduce the doubly-degenerate ES in the dominant singular values.
This qualitative difference highlights that deflation-based partial SVD approaches, which enforce orthogonality step-by-step, are essentially indispensable for capturing entanglement properties tied to topological degeneracy.

Beyond resolving qualitative topological features, the improved deflation algorithm also demonstrates robust quantitative accuracy.
As shown in Figs.~\ref{fig:vqsvd-demo-ladder}(e) and (f), the relative errors of the singular values systematically decrease as the number of layers $M$ increases, and this convergence behavior is independent of the underlying topological nature of the system.
Furthermore, the absolute error in the Schmidt gap, $\Delta=\xi_{2}-\xi_{1}$ shown in Figs.~\ref{fig:vqsvd-demo-ladder}(g) and (h), serves as a direct measure for resolving the topological degeneracy.
With a circuit depth of $M=4$ and $N=20$ deflation steps, the algorithm achieves an absolute error of $\sim 10^{-3}$ for both $J_{\perp}/J_{\parallel}=\pm 0.1$.
As can be clearly seen in the exact ES plotted in Fig.~\ref{fig:vqsvd-demo-ladder}(b), the entanglement gap separating the lowest degenerate levels from the higher-lying levels is on the order of $\mathcal{O}(1)$.
Compared to this gap scale, the achieved precision of $\sim 10^{-3}$ is more than sufficient to distinguish the topological degeneracy.
These results confirm that the improved deflation algorithm can efficiently construct accurate Schmidt decompositions for gapped systems, overcoming the limitations of global optimization strategies.


\bibliography{qsvd}

\end{document}